\documentclass[aps,pra,twocolumn,longbibliography,superscriptaddress]{revtex4-1}

\usepackage{amsmath}
\usepackage{amsthm}
\usepackage{amssymb}
\usepackage{mathtools} 
\usepackage{bm}
\usepackage{braket}
\usepackage{lineno}
\usepackage{pifont}
\usepackage{bbm}           
    
\newtheorem{theorem}{Theorem}
\newtheorem{corollary}{Corollary}
\newtheorem{lemma}{Lemma}
\newtheorem{proposition}{Proposition}
\newtheorem{definition}{Definition}

\usepackage{graphicx}
\usepackage{subfig}
\usepackage{booktabs}

\usepackage[hidelinks]{hyperref}
\usepackage{url}

\usepackage{caption}
\usepackage{adjustbox}

\DeclareMathOperator{\Tr}{Tr}

\DeclareMathOperator{\HQNN}{\mathcal{H}_{\mathsf{QNN}}}
\DeclareMathOperator{\HVQE}{\mathcal{H}_{\mathsf{VQE}}}
\DeclareMathOperator{\AQNN}{\mathcal{A}_{\mathsf{QNN}}}
\DeclareMathOperator{\AVQE}{\mathcal{A}_{\mathsf{VQE}}}

 \usepackage{xcolor}
\newcommand{\revise}[1]{\textcolor{black}{#1}}

\begin{document}
\title{Efficient measure for the expressivity of variational quantum algorithms}

\author{Yuxuan Du}
\email{duyuxuan123@gmail.com}
\affiliation{JD Explore Academy}
\author{Zhuozhuo Tu}
\email{zhtu3055@uni.sydney.edu.au}
\affiliation{School of Computer Science, The University of Sydney}

\author{Xiao Yuan}
\email{xiaoyuan@pku.edu.cn}
\affiliation{Center on Frontiers of Computing Studies, Department of Computer Science, Peking University, Beijing 100871, China}

\author{Dacheng Tao}
\email{dacheng.tao@gmail.com}
\affiliation{JD Explore Academy}

\begin{abstract}
The superiority of variational quantum algorithms (VQAs) such as quantum neural networks (QNNs) and variational quantum eigen-solvers (VQEs) heavily depends on the expressivity of the employed ans\"atze. Namely, a simple ansatz is insufficient to capture the optimal solution, while an intricate ansatz leads to the hardness of trainability. Despite its fundamental importance, an effective strategy of measuring the expressivity of VQAs remains largely unknown. Here, we exploit an advanced tool in statistical learning theory, i.e., covering number, to study the  expressivity of VQAs. Particularly, we first exhibit how the expressivity of VQAs with an arbitrary ans\"atze is upper bounded by the number of quantum gates and the measurement observable. We next explore the expressivity of VQAs on near-term quantum chips, where the system noise is considered. We observe an exponential decay of the expressivity with increasing circuit depth.  We also utilize the achieved expressivity to analyze the generalization of QNNs and the accuracy of VQE. We numerically verify our theory employing VQAs with different levels of expressivity. Our work opens the avenue for quantitative understanding of the expressivity of VQAs.         
\end{abstract}
   
 \maketitle    
 
\textit{Introduction.}---A paramount mission in quantum computing is devising learning protocols outperforming  classical methods~\cite{biamonte2017quantum,dunjko2018machine,harrow2017quantum}. Variational quantum algorithms (VQAs)~\cite{benedetti2019parameterized,bharti2021noisy,cerezo2020variational2,du2018expressive,endo2021hybrid} using parameterized quantum circuits~---~\textit{ans\"atze} and classical optimizers, serve as promising candidates to achieve this goal, especially in the noisy intermediate-scale quantum (NISQ) era~\cite{preskill2018quantum}.  Theoretical evidence has shown that VQAs may provide runtime speedups and enhanced generalization bounds for quantum information, quantum chemistry, and quantum machine learning (QML) tasks~\cite{Du_2021_grover,du2020learnability,huang2021information,shen2019information,wu2021expressivity}. Meanwhile, VQAs are flexible, which can adapt to restrictions imposed by NISQ devices such as qudits connectivity and shallow circuit depth. With this regard, great efforts have been dedicated to designing VQAs with varied ans\"atze to address different problems. Two important categories of existing VQAs include quantum neural networks (QNNs)~\cite{beer2020training,mitarai2018quantum,havlivcek2019supervised} and variational quantum eigen-solvers (VQEs)   ~\cite{wang2020quantum,peruzzo2014variational,kandala2017hardware}. Empirical studies have shown VQAs on near-term quantum devices achieving good performance for various tasks~\cite{google2020hartree,huang2020experimental,kandala2017hardware,zhu2019training}.                  
	
In parallel to the algorithm design, another central topic in the context of VQAs is exploring their learnability. A well study of this topic does not only allow us to understand the capabilities and limitations of VQAs with varied ans\"atze, but can also guide us to devise more powerful quantum protocols. As such, theoretical studies have attempted to exploit learnability of VQAs from distinct views.   Refs.~\cite{cerezo2020cost,mcclean2018barren,wang2020noise,zhang2020toward} have exhibited that the optimization of VQA suffers from  barren plateaus, where gradients information will be exponentially vanished with respected to the number of qudits and the circuit depth; Refs.~\cite{du2020learnability,sweke2020stochastic} have shown that more measurements, lower noise, and shallower circuit depth contribute to a better convergence  of QNNs with gradient descent optimizers; Refs.~\cite{abbas2020power,banchi2021generalization,bu2021statistical,du2020learnability,huang2021information,caro2020pseudo,funcke2021dimensional} have proven the generalization of QNNs with varied ans\"atze. {Ref.~\cite{poland2020no} has established quantum no-free lunch theorem of QNN and provided an apparently stronger lower bound than its classical counterpart}. {Very recently, Refs.~\cite{holmes2021connecting,sim2019expressibility,nakaji2021expressibility} connect the trainability and expressibility of VQAs, i.e., an ansatz exhibited with higher expressibility implies a flatter  loss landscapes and therefore will be harder to train. Hence, to ensure the power of VQAs, it is indispensable to develop an effective tool to measure the expressibility of VQAs. To this end, prior literature uses the unitary $t$-design to quantify the expressivity of VQAs~\cite{harrow2009random}. However, such a quantity is hard to calculate for a realistic quantum circuit and VQAs with well-designed ans\"atze may not obey the assumptions imposed by the unitary $t$-design~\cite{huggins2019towards,zhang2020toward}}. The above caveats motivate us to rethink: `\textit{Is there any effective and generic way to measure the expressivity of VQAs?}'     

Here, we provide a positive affirmation toward this question. Through connecting the expressivity with the model complexity, we leverage an advanced tool in statistical learning theory~---~covering number~\cite{vapnik2013nature}, to quantify the expressivity of VQAs.  We first exhibit that in the measure of the covering number, the upper bound of the expressivity for a given VQA yields $\mathcal{O}((N_{gt}\|O\|)^{d^{2k}N_{gt}})$, where $d$, $N_{gt}$, $k$, and $\|O\|$ refer to the dimension of \textit{qudit}, the number of trainable quantum gates, the largest number of qudits operated with a single quantum gate, and the operator norm of the observable $O$ used in the employed ans\"atze, respectively. With fixed $d$  and $\|O\|$, the expressivity of VQA can be well controlled by tuning $N_{gt}$ and $k$. Our second contribution is analyzing the expressivity of VQAs under the NISQ setting. When the quantum system noise is simulated by the depolarization channel, the expressivity is upper bounded by $\mathcal{O}((1-p)^{N_g}(N_{gt}\|O\|)^{d^{2k}N_{gt}})$, where $N_g$ is the total number of quantum gates (including both trainable and fixed ones) in the ans\"atze with $N_g\geq N_{gt}$ and $p$ is the depolarization rate. We further harness the derived expressivity to show that the generalization error bound of QNNs scales with $\tilde{\mathcal{O}}(\sqrt{N_{gt}}d^{k}/\sqrt{n})$, where $n$ is the number of training examples. This means that ans\"atze constituted by a large number of quantum gates request an increased number of training examples to ensure convergence. We believe that these observations may be of independent interest for the quantum machine learning community.

\begin{figure} 
\centering
\includegraphics[width=0.4\textwidth]{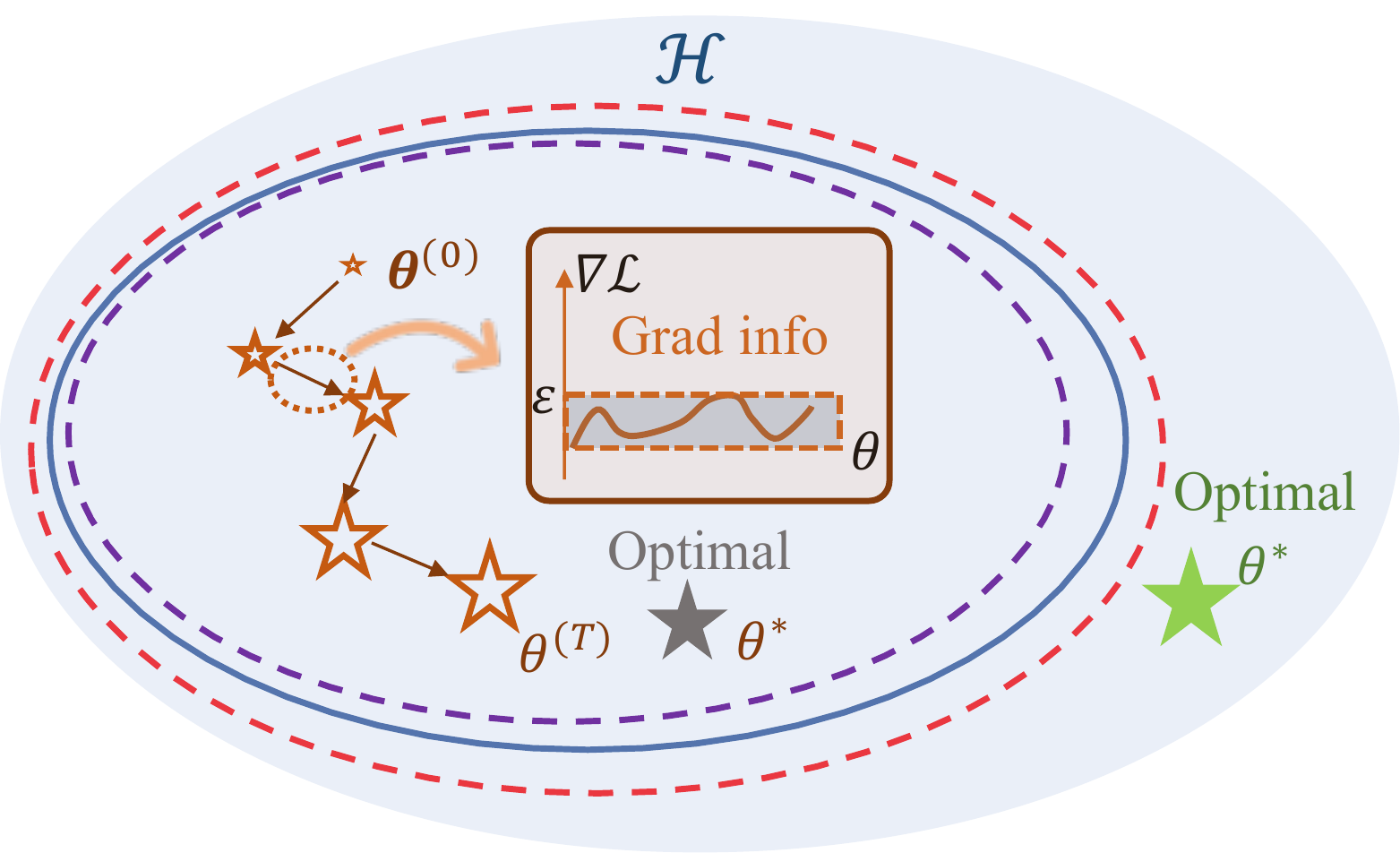}
\caption{\small{\textbf{Overview the expressivity of VQAs}. The expressivity of the employed ans\"atze of VQA rules its hypothesis space $\mathcal{H}$ (solid blue ellipse). When $\mathcal{H}$ has the modest size and covers the target concept (grey solid star), VQA can attain a good  performance. Conversely, when $\mathcal{H}$ fails to cover the target concept (green solid star), due to the limited expressivity, VQA achieves a poor performance. }}
\label{fig:scheme}
\end{figure}

\medskip       
\textit{Expressivity of VQA.}---We first review the working-flow of VQAs, which contains  an $N$-qudits quantum circuit and a classical optimizer. In the training stage, VQA follows an iterative manner to proceed optimization, where the optimizer continuously leverages the output of the quantum circuit to update trainable parameters of the adopted ansatz to minimize the predefined objective function $\mathcal{L}(\cdot)$. At the $t$-th iteration, the updating rule is $\bm{\theta}^{(t+1)}=\bm{\theta}^{(t)} - \eta \frac{\partial \mathcal{L}(h(\bm{\theta}^{(t)},O,\rho), c_1)}{\partial \bm{\theta}}$, where $\eta$ is the learning rate, $c_1\in\mathbb{R}$ is the target label, and $h(\bm{\theta}^{(t)},O,\rho)$ amounts to the output of the quantum circuit as elaborated below. Define $\rho\in\mathbb{C}^{d^N\times d^N}$ as the $N$-qudit input quantum state, $O\in\mathbb{C}^{d^N\times d^N}$ as the quantum observable,  $\hat{U}(\bm{\theta})=\prod_{l=1}^{N_g}\hat{u}_l(\bm{\theta}) \in\mathcal{U}(d^N)$ as the applied ansatz, i.e., $\bm{\theta}\in \Theta$ are trainable parameters living in the parameter space $\Theta$, $\hat{u}_l(\bm{\theta})\in\mathcal{U}(d^k)$ refers to the $l$-th quantum gate operated with at most $k$-qudits with $k\leq N$, and $\mathcal{U}(d^N)$ stands for the unitary group in dimension $d^N$. In general, $\hat{U}(\bm{\theta})$ is formed by $N_{gt}$ trainable gates and $N_g-N_{gt}$ fixed gates, e.g., $\Theta\subseteq [0, 2\pi)^{N_{gt}}$.  Under the above definitions, the explicit form of the output of the quantum circuit under the ideal scenario is 
\begin{equation}\label{eqn:unif-pqcs}
h(\bm{\theta}^{(t)},O,\rho):= 	\Tr\left(\hat{U}(\bm{\theta}^{(t)})^{\dagger}O\hat{U}(\bm{\theta}^{(t)}) \rho \right).
\end{equation}  
The gradients information ${\partial \mathcal{L}(h(\bm{\theta}^{(t)},O,\rho), c_1)}/{\partial \bm{\theta}}$ can be acquired via the parameter shift rule or other methods~\cite{mitarai2018quantum,schuld2019evaluating,stokes2020quantum}. The definition of $h(\bm{\theta}^{(t)},O,\rho)$ is generic. Here the  unitary $\hat{U}(\bm{\theta})$ covers many representative ans\"atze in QML and quantum chemistry, e.g., the hardware-efficient ansatz and unitary coupled-cluster ans\"atze~\cite{cerezo2020variational2}, and QNNs and VQEs can be effectively adapted to the form of $h(\bm{\theta}^{(t)},O,\rho)$ (see  following sections for details).

We now introduce the relationship between the expressivity and model complexity. In essence, the aim of VQAs is to find a good hypothesis $h^*(\bm{\theta},O,\rho)=\arg\min_{h(\bm{\theta},O,\rho) \in \mathcal{H}} \mathcal{L}(h(\bm{\theta},O,\rho) , c_1)$ that can well approximate the target concept, where $\mathcal{H}$ refers to the hypothesis space of VQA with 
\begin{equation}\label{eqn:general-hypo-clas}
	\mathcal{H} = \left\{ \Tr\left(\hat{U}(\bm{\theta})^{\dagger}O\hat{U}(\bm{\theta}) \rho \right)  \Big|  \bm{\theta}\in \Theta \right\}.
\end{equation}
An intuition about how the hypothesis space $\mathcal{H}$ affects the  performance of VQA is depicted in Fig.~\ref{fig:scheme}. When $\mathcal{H}$ has \textit{the modest size} and covers the target concepts, the estimated hypothesis could well approximate the target concept. By contrast,  when the complexity of $\mathcal{H}$ is too low, there exists a large gap between the estimated hypothesis and the target concept. Hence, to understand the expressivity of VQAs, it is highly demanded to devise an effective measure to evaluate the complexity of $\mathcal{H}$.     

Here we employ covering number, an advanced tool broadly used in statistical learning theory~\cite{vapnik2013nature}, {to bound the complexity of $\mathcal{H}$ and measure the expressivity of VQAs.}
\begin{definition}[Covering number]\label{def:cov-num}
The covering number $\mathcal{N}(\mathcal{U}, \epsilon, \|\cdot \|)$ denotes the least cardinality of any subset $\mathcal{V} \subset \mathcal{U}$ that covers $\mathcal{U}$ at scale $\epsilon$ with a norm $ \|\cdot \|$, i.e., $\sup_{A\in \mathcal{U}} \min_{B\in \mathcal{V}} \|A - B \|\leq \epsilon$. {Here we use this notion to measure the expressivity of VQAs.}
\end{definition}
 
The geometric interpretation of covering number is depicted in Fig.~\ref{fig:cov-numb}, which  refers to the minimum number of spherical balls with radius $\epsilon$ that are required to completely cover a given space with possible overlaps. {This notion has been employed to study other crucial topics in quantum physics such as Hamiltonian simulation~\cite{poulin2011quantum} and  entangled states~\cite{szarek2010often}}. {Note that $\epsilon$ is  a predefined hyper-parameter, i.e., a small constant with $\epsilon\in (0,1)$, and is independent with any factor~\cite{vapnik2013nature}. This convention has been  broadly adopted in the regime of machine learning to evaluate the model capacity of various  learning models~\cite{mohri2018foundations,vapnik2013nature}.} The following theorem shows the upper bound of $\mathcal{N}(\mathcal{H},  \epsilon, |\cdot|)$ whose proof is shown in Appendix~A.

\begin{theorem} \label{thm:main-cov-num}
For $0<\epsilon<1/10$, the covering number of the hypothesis space $\mathcal{H}$ in Eq.~(\ref{eqn:general-hypo-clas}) yields 
\begin{equation}\label{eqn:thm-1-key} 
 \mathcal{N}(\mathcal{H},  \epsilon, |\cdot|) \leq  \left(\frac{7 N_{gt}  \|O\| }{\epsilon} \right)^{d^{2k}N_{gt}},  
\end{equation}
where $\|O\|$ denotes the operator norm of $O$.    
\end{theorem}

It indicates that the most decisive factor, which controls the complexity of $\mathcal{H}$, is the employed quantum gates in $\hat{U}(\bm{\theta})$. This claim is ensured by the fact that the term $d^{2k}N_{gt}$ exponentially scales the complexity $\mathcal{N}(\mathcal{H},  \epsilon, |\cdot|)$. Meanwhile, the qudits count $N$ and the operator norm $\|O\|$ polynomially scale the complexity of $\mathcal{N}(\mathcal{H},  \epsilon, |\cdot|)$. These observations suggest a succinct and direct way to compare the expressivity of VQAs with differed  ans\"atze. {Moreover, different from prior works, we first prove that the expressivity of VQAs depends on the type of quantum gates (denoted by the term $k$). Since it is a long standing problem of proving that the expressivity of VQAs depends on the structure information of ansatz such as the location of different quantum gates and the types of the employed quantum gates, our result makes a concrete progress toward this goal. It is noteworthy that our results do not only indicate a general scaling behavior of the model's expressivity, but also  provide a practical guidance of designing VQA-based models. In Appendix H, we elaborate how to combine the achieved theoretical results with  \textit{structural risk minimization} to enhance the learning performance of VQA-based models~\cite{mohri2018foundations}.}   
 
  \begin{figure}[t]
	\centering
\includegraphics[width=0.25\textwidth]{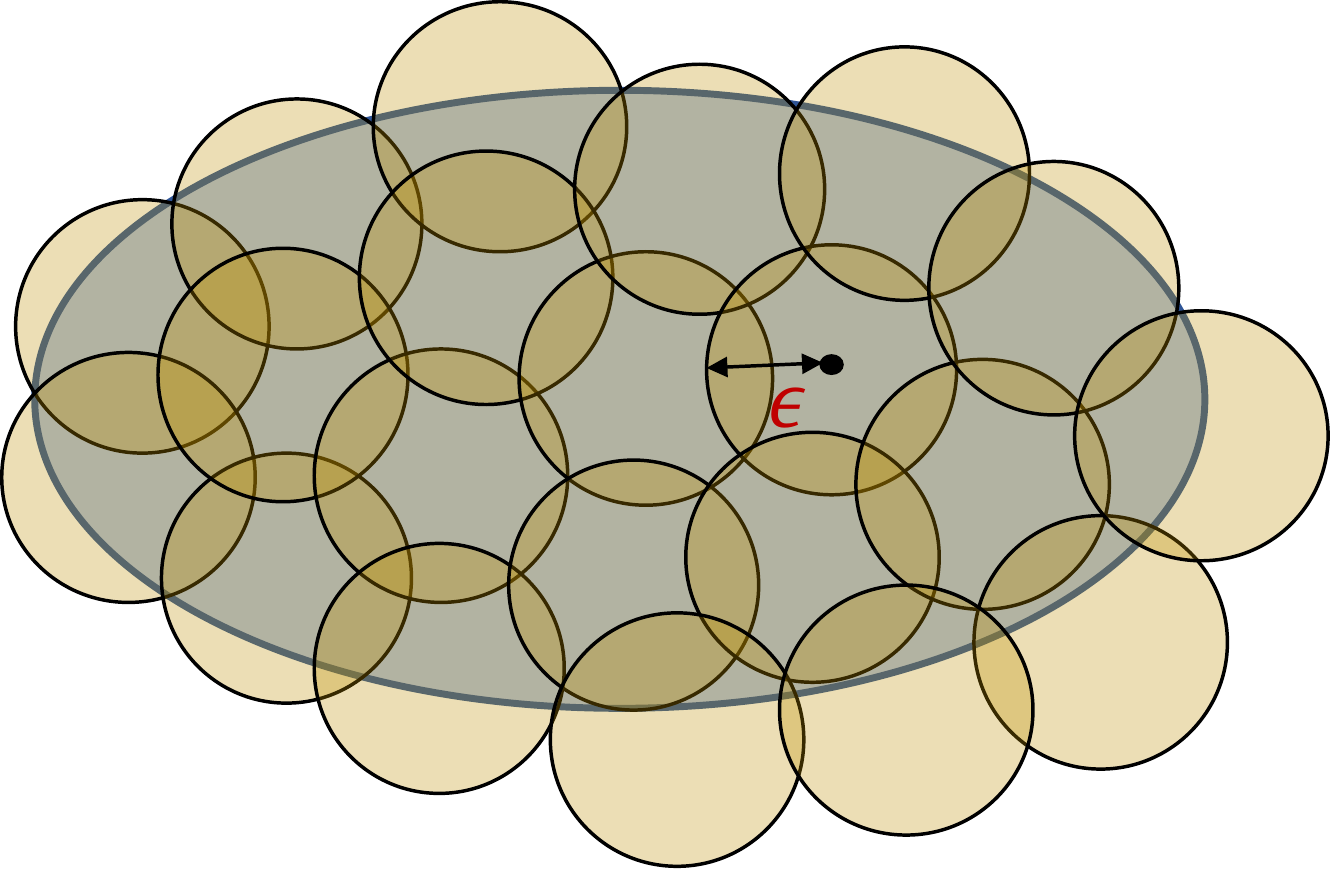}
\caption{\small{\textbf{The geometric intuition of covering number}. Covering number concerns the minimum number of spherical balls with radius $\epsilon$ that occupies the whole space.  }}
\label{fig:cov-numb}
\end{figure}

\textbf{Remark.}  Theorem~\ref{thm:main-cov-num} is \textit{ubiquitous} and \textit{do not rely on the assumption of the unitary $t$-design}, which differs from~\cite{holmes2021connecting}. Moreover, the qubit-based VQAs are a special case of our results with $d=2$. {We also study the tightness of the bound In Appendix I.}  
            
\medskip
We next consider how the expressivity, or equivalently covering number, of VQA varies when noise $\mathcal{E}(\cdot)$ is considered. Under this scenario, the hypothesis space of VQA in Eq.~(\ref{eqn:general-hypo-clas}) transforms to 
\begin{equation}\label{eqn:general-hypo-clas-NISQ}
	\widetilde{\mathcal{H}} = \left\{ \Tr\left( O \mathcal{E} \left(\hat{U}(\bm{\theta}) \rho \hat{U}(\bm{\theta})^{\dagger}\right) \right) \Big|  \bm{\theta}\in \Theta \right\}.
\end{equation}
The expressivity of noisy VQAs is summarized in Proposition \ref{prop:cov-vqa-noise}, whose proof is provided in Appendix B. 
\begin{proposition}\label{prop:cov-vqa-noise}
Following notations and conditions in Theorem \ref{thm:main-cov-num}, the covering number of $\widetilde{\mathcal{H}}$ in Eq.~(\ref{eqn:general-hypo-clas-NISQ}) satisfies $\mathcal{N}(\widetilde{\mathcal{H}},  \epsilon, |\cdot|) \leq  2\|O\|  \left(\frac{7 N_{gt}  }{\epsilon} \right)^{d^{2k}N_{gt}}$. If $\mathcal{E}(\cdot)$ further refers to the depolarization channel  $\mathcal{E}_{p}(\rho) = (1-p)\rho + p\mathbb{I}/d^N$ that is applied to each quantum gate, the covering number of  $\widetilde{\mathcal{H}}$ satisfies  $\mathcal{N}(\widetilde{\mathcal{H}},  \epsilon, |\cdot|) \leq 	(1-p)^{N_g}   \left(\frac{7 N_{gt}  \|O\| }{\epsilon} \right)^{d^{2k}N_{gt}}$. 
\end{proposition}
Proposition~\ref{prop:cov-vqa-noise} indicates the following insights. First, the expressivity of VQAs under  general system noise setting can not be better than their ideal cases, since for both cases, the term $N_{gt}d^{2k}$ exponentially scales the expressivity of $\widetilde{\mathcal{H}}$. Second,  the upper bounds about the expressivity given in Eq.~(\ref{eqn:thm-1-key}) and  Proposition \ref{prop:cov-vqa-noise} suggest that quantum noise cannot increase the expressivity of VQA compared with its ideal case~\footnote{{We remark that to obtain a general result that covers any type of noise, a relatively loose relaxation technique is used to infer $\mathcal{N}(\tilde{\mathcal{H}},  \epsilon, |\cdot|)$. This leads to a different scaling behavior in term of $\|O\|$ comparing with the ideal and depolarizing cases.}}. Additionally, in the worst scenario where the depolarization noise is considered, the factor $(1-p)^{N_g}$ shrinks the expressivity of $\mathcal{H}$. These insights enables us to compare the expressivity of different VQAs in the NISQ scenario. Meanwhile, the system noise may \textit{forbid us} to devise high-expressive VQAs, due to the term $(1-p)^{N_g}$. Hence, integrating error mitigation techniques with VQAs is  desired~\cite{cai2020mitigating,du2020quantum,mcardle2019error,mcclean2020decoding,strikis2020learning,sun2021mitigating}.    

 \medskip
To better elucidate how  Theorem \ref{thm:main-cov-num} and Proposition \ref{prop:cov-vqa-noise} contribute to concrete quantum learning tasks, in the following, we separately explore the expressivity of QNNs and VQEs, as two crucial subclasses of VQAs.

\textit{Expressivity of quantum neural networks.}---The aim of machine learning is devising an algorithm $\mathcal{A}$ so that given a training dataset $S=\{ (\bm{x}^{(i)}, \bm{y}^{(i)})\}_{i=1}^n$ sampled from the domain $\mathcal{X}\times \mathcal{Y}$, $\mathcal{A}$ can use $S$ to infer a hypothesis $h_{\mathcal{A}(S)}^*(\cdot)$ from its hypothesis space to \textit{minimize} the expected risk $\mathcal  R(\mathcal{A}(S)) = \mathbb{E}_{\bm{x}, \bm{y}}(\ell(h_{\mathcal{A}(S)}(\bm{x}), \bm{y}))$~\cite{kawaguchi2017generalization}, where the randomness is taken over $\mathcal{A}$ and $S$, and $\ell$ refers to the designated loss function. Since the probability distribution behind data space $\mathcal{X}\times \mathcal{Y}$ is generally inaccessible, the minimization of $\mathcal R(\mathcal{A}(S))$ becomes intractable. To tackle this issue, an alternative way of inferring  $h^*(\cdot)$ is minimizing the empirical risk $\hat{\mathcal R}_S(\mathcal{A}(S)) = \frac{1}{n} \sum_{i=1}^n \ell(h_{\mathcal{A}(S)}(\bm{x}^{(i)}), \bm{y}^{(i)})$. 

When QNN is employed to implement $\mathcal{A}$ (as denoted by $\AQNN$) to minimize $\hat{\mathcal R}_S(\mathcal{A}(S)) $, its paradigm can be cast into Eq.~(\ref{eqn:unif-pqcs}). Given the classical example $\bm{x}^{(i)}$, QNN first  prepares an input quantum state $\rho_{\bm{x}^{(i)}}\in \mathbb{C}^{2^N\times 2^N}$ that loads $\bm{x}^{(i)}$ adopting various encoding methods~\cite{benedetti2019parameterized,larose2020robust}. Once the state $\rho_{\bm{x}^{(i)}}$ is prepared, the ansatz $\hat{U}(\bm{\theta}^{(t)})$ is applied to this state, followed by a predefined  quantum measurement $O$. To this end, the explicit form of a hypothesis for QNN is
\begin{equation}\label{eqn:hypothsis-QNN}
	h_{\AQNN(S)}(\bm{x}^{(i)}) = \Tr\left(\hat{U}(\bm{\theta})^{\dagger}O \hat{U}(\bm{\theta}) \rho_{\bm{x}^{(i)}} \right),
\end{equation}  
where $\AQNN(S)=\bm{\theta}\in 
 \Theta$ represents the updated parameters. Since the parameter space $\Theta$ is bounded, the hypothesis space of QNN follows 
\begin{equation}\label{eqn:hypo-space-QNN}
	\HQNN = \left\{h_{\AQNN(S)}(\cdot) \Big| \bm{\theta}\in \Theta \right\}. 
\end{equation}        

The explicit form of $\HQNN$ allows us to directly make use of Theorem \ref{thm:main-cov-num} and Proposition \ref{prop:cov-vqa-noise} to analyze the expressivity of various QNNs. To facilitate understanding, in Appendix D,  we analyze the expressivity of QNNs with typical ans\"atze such as hardware-efficient and tensor-network based ans\"atze.

Here we also explore the \textit{generalization error} of QNNs, an important concept in quantum learning theory, which    explains that when and how minimizing $\hat{\mathcal R}_S(\mathcal A(S))$ is a sensible approach to minimizing $\mathcal  R(\mathcal A(S))$ by analyzing the upper bound of $\mathcal{R}(\mathcal{A}(S))-\hat{\mathcal{R}}_S(\mathcal{A}(S))$. The generalization bound can be effectively derived when the complexity of hypothesis space is accessible~\cite{mohri2012foundations}. Hence, we use Theorem \ref{thm:main-cov-num} to obtain the following claim whose proof is given in Appendix C.   
\begin{theorem}\label{thm:generalization-QNN}
Assume the loss $\ell$ is $L_1$-Lipschitz and upper bounded by $C_1$. Following notations in Eq.~(\ref{eqn:hypo-space-QNN}), for $0<\epsilon<1/10$, with probability at least $1-\delta$ with $\delta\in (0, 1)$, 
\[
	 \mathcal{R}(\mathcal{A}(S)) - \hat{\mathcal{R}}_S(\mathcal{A}(S))  \leq  \tilde{\mathcal{O}}\bigg(\frac{8L_1+ C_1 + 24L_1d^{k}\sqrt{N_{gt}}}{\sqrt{n}}\bigg).       
\]
\end{theorem}   
{The employed assumption is very mild, since the loss functions adopted in QNNs are generally Lipschitz continuous and can be upper bounded by a constant $C_1$. This property has been broadly employed to understand the capability of QNNs~\cite{abbas2020power,du2020learnability,huang2021information,marrero2020entanglement,mcclean2018barren,sweke2019stochastic}.} The achieved results provide three-fold implications. First, the generalization bound has an exponential dependence with the term $k$ and the \textit{sublinear} dependence with the number of trainable quantum gates $N_{gt}$. This observation reveals an Occam’s razor principle in the quantum version~\cite{blumer1987occam}, where parsimony of the output hypothesis implies predictive power. Second, increasing the size of training examples $n$ contributes to an improved generalization bound. This outcome requests us to involve more training data to optimize intricate ans\"atze. Last, the sublinear dependence of $N_{gt}$ may limit our result to accurately assess the generalization ability for the over-parameterized QNNs  ~\cite{larocca2021theory}. A future work is to integrate our results with deep learning theory that focuses on over-parameterized model to derive a tighten bound ~\cite{canatar2021spectral}. All of these implications can be employed as guidance to design powerful QNNs.        

\begin{figure}
\centering
\includegraphics[width=0.48\textwidth]{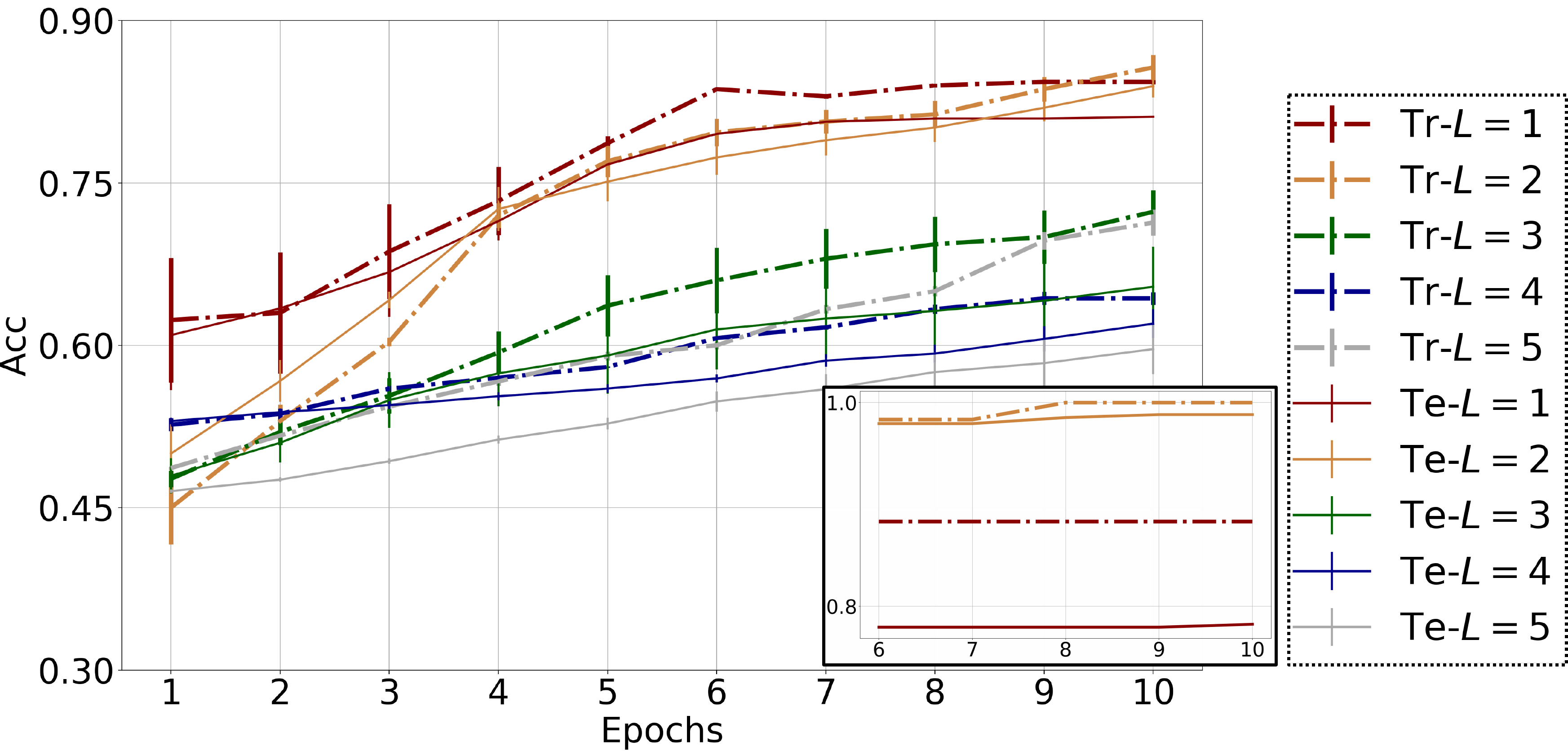}	
\caption{\small{\textbf{Simulation results of QNN with the varied layer number $L$.} The label Tr-$L=a$ (or Te-$L=a$) refers to the train (or test) accuracy of QNN with layer number $L=a$. The outer (or inner) plot shows the statistical (top-1) results of QNNs with {$L=\{1,2,...5\}$ ($L=\{1,2\}$)}. Vertical bars   refer to the variance of obtained results.}}  
\label{fig:sim-QNN}
\end{figure}   

We conduct numerical simulations to validate our theoretical results. Specifically, we apply QNNs to accomplish the binary classification task on the synthetic dataset $S$. The construction of $S$ follows~\cite{havlivcek2019supervised}, where the dataset consists of $400$ examples, and   the feature dimension of $\bm{x}^{(i)}$ is $7$ and the corresponding label $y^{(i)}\in\{0,1\}$ is binary~$\forall i\in [400]$.  At the data preprocessing stage,  $S$ is divided into the training set and the test set with size $60$ and $340$, respectively. The implementation of QNN is as follows. The qubit encoding method is used to load $\rho_{\bm{x}^{(i)}}$. The the layer number $L$ of $\hat{U}(\bm{\theta})=\prod_{l=1}^L U(\bm{\theta}^{l})$ is varied from $1$ to $5$. Notably, when $L\geq 2$, the target concept is contained in $\HQNN$.  We repeat each setting with $5$ times to gain statistical results.  See Appendix E for construction details.  

The simulation results are exhibited in Fig.~\ref{fig:sim-QNN}. Although $\HQNN$ with the layer number $L\in\{2,3,4,5\}$ covers the target concept, the trainability, as reflected by the training accuracy in the outer plot, becomes deteriorating with respect to increased $L$. This result echoes with Theorem \ref{thm:main-cov-num} in the sense that high expressivity implies poor trainability. Moreover, the discrepancy between train and test accuracy of QNN becomes large, especially for $L=5$. This result accords with Proposition \ref{prop:cov-vqa-noise} such that higher expressivity results in  larger generalization error. Eventually, in conjunction with the inner and outer plots with $L=1$, we conclude that when the expressivity of $\HQNN$ is too small, which excludes the target concept, the training of QNN is stable but with a high empirical risk. The performance of QNNs in the NISQ case is deferred to Appendix E.

\textit{Expressivity of variational quantum eigen-solvers.}---A central task in quantum chemistry is designing an efficient algorithm to estimate low-lying eigenstates and corresponding eigenvalues of an input Hamiltonian~\cite{mcardle2020quantum}. Variational quantum eigen-solvers (VQEs), denoted by $\AVQE$, are the most popular protocols to reach this goal in the NISQ era~\cite{peruzzo2014variational}, owing to  their capability and flexibility. The training of VQE also adopts the iterative manner and each iteration includes two steps. Initially, VQE applies an ansatz $U(\bm{\theta})=\prod_{l=1}^LU_l(\bm{\theta})$ to a fixed $N$-qubit quantum state $\rho_0=(\ket{0}\bra{0})^{\otimes N}$, followed by measuring the Hamiltonian $H$ to collect the classical outputs. Then, the classical optimizer utilizes the output information to update $\bm{\theta}$ via gradient descent method to minimize $\Tr(H U(\bm{\theta})\rho_0 U(\bm{\theta})^{\dagger})$. The hypothesis space of VQE can be exactly formulated by Eq.~(\ref{eqn:general-hypo-clas}), i.e., 
\begin{equation}\label{eqn:hypo-space-VQE}\nonumber
\HVQE = \{h_{\AVQE(H)}(\rho_0)  :=\Tr(H  U(\bm{\theta})\rho_0  U(\bm{\theta})^{\dagger}) | \bm{\theta}\in \Theta \}. 
\end{equation}

The form of $\HVQE$ enables us to efficiently measure the expressivity of an arbitrary ansatze used in VQEs by using Theorem \ref{thm:main-cov-num} and Proposition \ref{prop:cov-vqa-noise}. For concreteness, we quantify the expressivity of unitary coupled-cluster ans\"atze truncated up to single and double excitations (UCCSD)~\cite{cao2019quantum} whose proof is given in Appendix F.             
 
\begin{corollary}\label{coro:VQE} 
	Under the ideal setting, the covering number of VQE with UCCSD is upper bounded by $\mathcal{N}(\HVQE,  \epsilon, |\cdot|) \leq  \mathcal{O}(\frac{7N^5 \|H\|}{\epsilon} )^{d^{2k}N^5})$. When the system noise is considered and simulated by the depolarization channel, the corresponding covering number is upper bounded by $\mathcal{N}(\widetilde{\HVQE},  \epsilon, |\cdot|) \leq  \mathcal{O}((1-p)^{N^5}(\frac{7N^5 \|H\|}{\epsilon} )^{d^{2k}N^5})$.
\end{corollary}

\begin{figure}
	\centering
\includegraphics[width=0.45\textwidth]{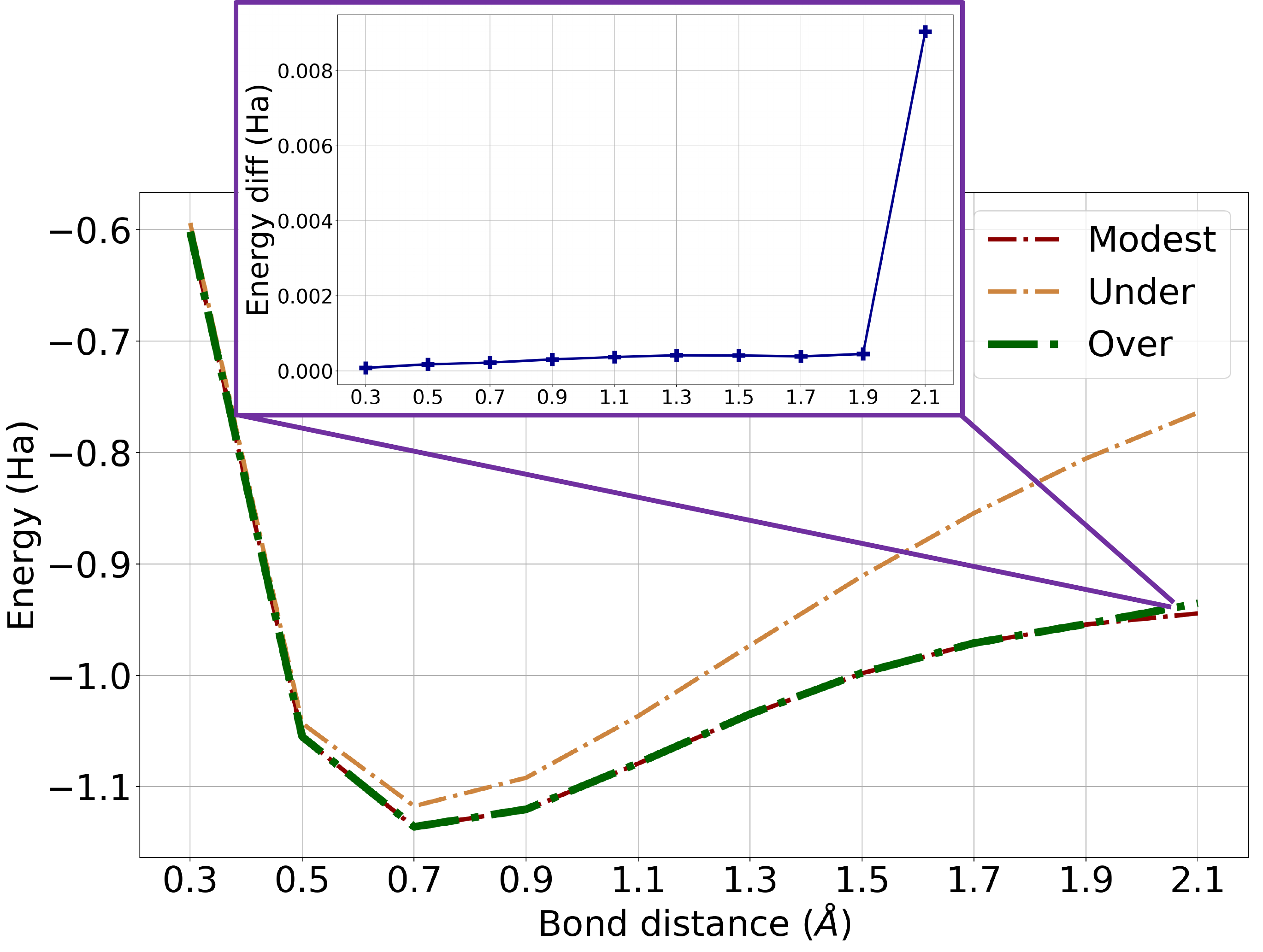}
\caption{\small{\textbf{Simulation results of VQE.} The labels `Under', `Modest', and `Over' refer to the estimated energy of VQEs when the employed ans\"atze have the restricted, modest, and overwhelming expressivity, respectively. \revise{`Ha' (Hartrees) and `\r{A}' (Angstroms) refer to the units for energy  and the bond lengths}. The inner plot shows the energy gap of VQEs in `Over' and `Modest' cases.}} 
\label{fig:VQE-sim}
\end{figure}

We conduct numerical simulations to explore how the expressivity of ans\"atze affects the performance of VQEs. In particular, we apply VQEs using three different ans\"atze with insufficient, modest, and overwhelming expressivity~\footnote{The separated expressivity of different ansatze is completed by controlling the involved number of quantum gates, supported by Theorem~\ref{thm:main-cov-num}. See Appendix G for construction details.}. and  estimate the ground state energy of Hydrogen molecule with varied bond length ranging from $0.3 \mathring{A}$ to $2.1\mathring{A}$. As shown in  Fig.~\ref{fig:VQE-sim}, VQE with the restricted ans\"atze performs worse than the modest and overwhelming ans\"atze, where there exists an apparent energy gap when the bond length is larger than $0.7 \mathring{A}$. Furthermore, although VQEs with the modest and overwhelming ans\"atze demonstrate  similar behavior in all bond lengths,  the former always outperforms the latter as shown in the inner plot of Fig.~\ref{fig:VQE-sim}. The collected results indicate that too limited or too redundant expressivity of the employed ans\"atze may prohibit the trainability of VQE ({in Appendix G, we numerically evidence that such a claim still holds when the learning rate is allowed to be adaptive}).

\medskip 
\textit{Discussion and conclusions.}---We devise an efficient measure to quantify the expressivity of VQAs, including QNNs and VQEs, controlled by the qudits count, the involved quantum gates in ans\"atze, the operator norm of the observable, and the system noise. Compared with the prior study~\cite{holmes2021connecting},  our results allow a succinct and direct way to compare the expressivity of different ans\"atze and devise novel ans\"atze. Our work mainly concentrates on the upper bounds of expressivity, whereas a promising research direction is to derive lower bounds and tighten the expressivity quantity. The developed tool here can be extended to analyze generalization ability of other advanced QNNs such as quantum convolutional neural networks. Besides, considering that generalization bounds can be used to design an ansatz with good learning performance via the framework of structural risk minimization, it is intrigued to use our results as a theoretical guidance to devise advanced QNNs.   Another crucial research direction is exploring explicit quantification of the `modest expressivity' of VQAs. A deep understanding of this issue contributes to integrate various  NISQ-oriented techniques such as error mitigation and quantum circuit architecture design techniques to boost the VQAs performance. 
 
 \begin{acknowledgements}
Y. Du thanks V. Schmiesing for providing valuable discussions. X. Yuan is supported by the National Natural Science Foundation of China Grant No.~12175003.
\end{acknowledgements}

\clearpage
\newpage

\renewcommand\thefigure{\thesection.\arabic{figure}}     
\appendix 
 
\onecolumngrid

\section{Proof of Theorem 1}\label{append:main-thm1}
The proof of Theorem 1 employs the definition of the operator norm.  
\begin{definition}[Operator norm]\label{def:opt-norm}
	Suppose A is an $n\times n$ matrix. The operator norm $A$ is defined as
	\begin{equation}
		\|A\| = \sup_{\|\bm{x}\|_2=1, \bm{x}\in\mathbb{C}^n} \|A\bm{x}\|.
	\end{equation}
	Alternatively, $\|A\|=\sqrt{\lambda_{1}(AA^{\dagger})}$, where $\lambda_i(AA^{\dagger})$ is the $i$-th largest eigenvalue of the matrix $AA^{\dagger}$.
\end{definition}

Besides the above definition, the proof of Theorem 1  leverages the the following two lemmas. In particular, The first lemma enables us to employ the covering number of one metric space $(\mathcal{H}_1, d_1)$ to bound the covering number of an another metric space $(\mathcal{H}_2, d_2)$. 

\begin{lemma}[Lemma 5, \cite{Barthel2018fundamental}]\label{lem:cov-num-two-space}
	Let $(\mathcal{H}_1, d_1)$ and $(\mathcal{H}_2, d_2)$ be metric spaces and $f:\mathcal{H}_1 \rightarrow  \mathcal{H}_2 $ be bi-Lipschitz such that 
	\begin{equation}\label{eqn:lem-cov-num-two-space-1}
		d_2(f(\bm{x}), f(\bm{y})) \leq K d_1(\bm{x}, \bm{y}), ~\forall \bm{x}, \bm{y} \in \mathcal{H}_1,
	\end{equation}
	and 
	\begin{equation}\label{eqn:lem-cov-num-two-space-2}
		d_2(f(\bm{x}), f(\bm{y})) \geq k d_1(\bm{x}, \bm{y}), ~\forall \bm{x}, \bm{y} \in \mathcal{H}_1~\text{with}~ d_1(\bm{x}, \bm{y}) \leq r.
	\end{equation}
Then their covering numbers obey
\begin{equation}\label{eqn:lem-cov-num-two-space-3}
	\mathcal{N}(\mathcal{H}_1, 2\epsilon /k, d_1 ) \leq \mathcal{N}(\mathcal{H}_2,  \epsilon, d_2   ) \leq  \mathcal{N}(\mathcal{H}_1, \epsilon /K, d_1 ),
\end{equation}
where the left inequality requires $\epsilon \leq kr/2$.
\end{lemma}

The second lemma presents the covering number of the operator  group 
\begin{equation}\label{eqn:def-opt-group}
	\mathcal{H}_{circ}:=\left\{\hat{U}(\bm{\theta})^{\dagger}O\hat{U}(\bm{\theta})| \bm{\theta}\in \Theta \right\},
\end{equation}
where $\hat{U}(\bm{\theta}) = \prod_{i=1}^{N_g}\hat{u}_i(\bm{\theta}_i)$ and only $N_{gt}\leq N_g$ gates in $U(\bm{\theta})$ are trainable. \revise{The detailed proof is deferred to Appendix  \ref{append:subsec-cov-num-Lemma2}}.  
\begin{lemma}\label{lem:thm-cov-num-opt-set}
	Following notations in Theorem 1, suppose that the employed $N$-qubit Ans\"atze containing in total $N_g$ gates with $N_g>N$, each gate $\hat{u}_i(\bm{\theta})$ acting on at most $k$ qudits, and $N_{gt}\leq N_g$ gates in $U(\bm{\theta})$ are trainable. The $\epsilon$-covering number for the operator  group $\mathcal{H}_{circ}$ in Eq.~(\ref{eqn:def-opt-group}) with respect to the operator-norm distance obeys
\begin{equation}
  \mathcal{N}(\mathcal{H}_{circ}, \epsilon , \|\cdot\|) \leq \left(\frac{7 N_{gt}  \|O\| }{\epsilon} \right)^{d^{2k}N_{gt}},
\end{equation}
where $\|O\|$ denotes the operator norm of $O$.
\end{lemma}

We are now ready to present the proof of Theorem 1.
\begin{proof}[Proof of Theorem 1]
The intuition of the proof is as follows. Recall the definition of the hypothesis space $\mathcal{H}$ in Eq.~(2) and Lemma \ref{lem:cov-num-two-space}. When $\mathcal{H}_1$ refers to the hypothesis space $\mathcal{H}$ and $\mathcal{H}_2$ refers to the unitary group $\mathcal{U}(d^N)$, the upper bound of the covering number of $\mathcal{H}$, i.e., $\mathcal{N}(\mathcal{H}_1, d_1, \epsilon )$, can be derived by first quantifying $K$  Eq.~(\ref{eqn:lem-cov-num-two-space-1}), and then interacting with $\mathcal{N}(\mathcal{H}_{circ}, \epsilon, \|\cdot\|)$ in Lemma \ref{lem:thm-cov-num-opt-set}. Under the above observations, in the following, we analyze the upper bound of the covering number $\mathcal{N}(\mathcal{H},  \epsilon, |\cdot|)$.
  
 We now derive the Lipschitz constant $K$ in Eq.~(\ref{eqn:lem-cov-num-two-space-1}), as the precondition to achieve the upper bound of $\mathcal{N}(\mathcal{H},  \epsilon, |\cdot |)$. Define $\hat{U}\in \mathcal{U}(d^N)$ as the employed Ans\"atze composed of $N_g$ gates, i.e., $\hat{U}=\prod_{i=1}^{N_g} \hat{u}_l$.  Let $\hat{U}_{\epsilon}$ be the quantum circuit where each of the $N_g$ gates is replaced by the nearest element in the covering set. The relation between the distance $d_2(\Tr(\hat{U}_{\epsilon}^{\dagger}O\hat{U}_{\epsilon}\rho), \Tr(\hat{U}^{\dagger}O\hat{U} \rho))$ and the distance $d_1(\hat{U}_{\epsilon}, \hat{U})$ yields  
\begin{eqnarray}
	&& d_2(\Tr(\hat{U}_{\epsilon}^{\dagger}O\hat{U}_{\epsilon}\rho), \Tr(\hat{U}^{\dagger}O\hat{U} \rho)) \nonumber\\
=   && |\Tr(\hat{U}_{\epsilon}^{\dagger}O\hat{U}_{\epsilon}\rho) - \Tr(\hat{U}^{\dagger}O\hat{U} \rho) | \nonumber\\
	= && \left|\Tr\left((\hat{U}_{\epsilon}^{\dagger}O\hat{U}_{\epsilon}  -  \hat{U}^{\dagger}O\hat{U} ) \rho \right) \right| \nonumber\\
	\leq && \left\|\hat{U}_{\epsilon}^{\dagger}O\hat{U}_{\epsilon}  -  \hat{U}^{\dagger}O\hat{U} \right\|\Tr(\rho) \nonumber\\
	= && d_1(\hat{U}_{\epsilon}^{\dagger}O\hat{U}_{\epsilon}, \hat{U}^{\dagger}O\hat{U}), 
\end{eqnarray}
where the first equality comes from the explicit form of hypothesis, the first inequality uses the Cauchy-Schwartz inequality, and the last inequality employs $\Tr(\rho)=1$ and  
 \begin{equation}\label{eqn:dist-cov}
	\left\|\hat{U}_{\epsilon}^{\dagger}O\hat{U}_{\epsilon}  -  \hat{U}^{\dagger}O\hat{U} \right\| = d_1(\hat{U}_{\epsilon}^{\dagger}O\hat{U}_{\epsilon}, \hat{U}^{\dagger}O\hat{U}).
\end{equation}
 
The above equation indicates $K=1$. Combining the above result with   Lemma \ref{lem:cov-num-two-space} (i.e., Eq.~(\ref{eqn:lem-cov-num-two-space-1})) and Lemma \ref{lem:thm-cov-num-opt-set}, we obtain
\begin{equation}
	\mathcal{N}(\mathcal{H},  \epsilon,  |\cdot |)  \leq   \mathcal{N}(\mathcal{H}_{circ}, \epsilon, \|\cdot\|) \leq  \left(\frac{7 N_{gt}  \|O\| }{\epsilon} \right)^{d^{2k}N_{gt}}.
\end{equation} 
  
This relation ensures 
\begin{equation}
	\mathcal{N}(\mathcal{H},  \epsilon, |\cdot|) \leq \left(\frac{7 N_{gt}  \|O\| }{\epsilon} \right)^{d^{2k}N_{gt}}.
\end{equation}

\end{proof}

\subsection{Proof of Lemma \ref{lem:thm-cov-num-opt-set}}\label{append:subsec-cov-num-Lemma2}
The proof of Lemma \ref{lem:thm-cov-num-opt-set} exploits the following result.
\begin{lemma}[Lemma 1, \cite{Barthel2018fundamental}]\label{lem:cov-fund}
For $0<\epsilon<1/10$, the $\epsilon$-covering number for the unitary group $U(d^k)$ with respect to the operator-norm distance in Definition \ref{def:opt-norm} obeys
\begin{equation}
	\left(\frac{3}{4\epsilon} \right)^{d^{2k}} \leq \mathcal{N}(U(d^k), \epsilon, \|\cdot\|) \leq \left(\frac{7}{\epsilon} \right)^{d^{2k}}.
\end{equation}	
\end{lemma}
 \begin{proof}[Proof of Lemma \ref{lem:thm-cov-num-opt-set}]
 
 The goal of Lemma \ref{lem:thm-cov-num-opt-set} is to measure the covering number of the operator group $\mathcal{H}_{circ}= \{\hat{U}(\bm{\theta})^{\dagger}O\hat{U}(\bm{\theta})| \bm{\theta}\in \Theta  \}$ in Eq.~(\ref{eqn:def-opt-group}), where the trainable unitary $\hat{U}(\bm{\theta}) = \prod_{i=1}^{N_g}\hat{u}_i(\bm{\theta}_i)$ consists of $N_{gt}$ trainable gates and  $N_g-N_{gt}$ fixed gates. To achieve this goal, we consider a fixed $\epsilon$-covering $\mathcal{S}$ for the set $\mathcal{N}(U(d^k), \epsilon, \|\cdot\|)$ of all possible gates and define the set 
 \begin{equation}
 	\tilde{\mathcal{S}}:=\left\{\prod_{i\in\{N_{gt}\}}\hat{u}_i(\bm{\theta}_i)\prod_{j\in\{N_g-N_{gt}\}}\hat{u}_j\Big| \hat{u}_i(\bm{\theta}_i)\in \mathcal{S}\right\},
 \end{equation}
 where $\hat{u}_i(\bm{\theta}_i)$ and $\hat{u}_j$ specify to the trainable and fixed quantum gates in the employed Ans\"atze, respectively. Note that for any circuit $\hat{U}(\bm{\theta}) = \prod_{i=1}^{N_g}\hat{u}_i(\bm{\theta}_i)$, we can always find a $\hat{U}_{\epsilon}(\bm{\theta})\in\tilde{\mathcal{S}}$ where each $\hat{u}_i(\bm{\theta}_i)$ of trainable gates is replaced with the nearest element in \textit{the covering set $\mathcal{S}$}, and the discrepancy $\|\hat{U}(\bm{\theta})^{\dagger}O\hat{U}(\bm{\theta})-\hat{U}_{\epsilon}(\bm{\theta})^{\dagger}O\hat{U}_{\epsilon}(\bm{\theta})\|$ satisfies
 \begin{eqnarray}\label{eqn:append-proof-lemma2-1}
  && \|\hat{U}(\bm{\theta})^{\dagger}O\hat{U}(\bm{\theta})-\hat{U}_{\epsilon}(\bm{\theta})^{\dagger}O\hat{U}_{\epsilon}(\bm{\theta})\| \nonumber\\
  \leq &&  \|\hat{U} -\hat{U}_{\epsilon}\|   \|O\| \nonumber\\
  \leq &&    N_{gt} \|O\|  \epsilon,
 \end{eqnarray} 
where the first inequality uses the triangle inequality, and the second inequality follows from $ \|\hat{U} -\hat{U}_{\epsilon}\| \leq N_{gt}\epsilon$.

Therefore, by Definition 1, we know that $\tilde{\mathcal{S}}$ is a $N_{gt} \|O\| \epsilon$-covering set for $\mathcal{H}_{circ}$. Recall that the upper bound in Lemma \ref{lem:cov-fund} gives $|\mathcal{S}|\leq\left(\frac{7}{\epsilon} \right)^{d^{2k}}$. Since there are $|\mathcal{S}|^{N_{gt}}$ combinations for the gates in $\tilde{\mathcal{S}}$, we have $|\tilde{\mathcal{S}}|\leq\left(\frac{7}{\epsilon} \right)^{d^{2k}N_{gt}}$  and the covering number for $\mathcal{H}_{circ}$ satisfies 
\begin{equation}
\mathcal{N}(\mathcal{H}_{circ},  N_{gt}  \|O\|  \epsilon, \|\cdot\|) \leq  \left(\frac{7 }{\epsilon} \right)^{d^{2k}N_{gt}}.  	
\end{equation}
An equivalent representation of the above inequality is
\begin{equation}
	\mathcal{N}(\mathcal{H}_{circ},   \epsilon, \|\cdot\|) \leq  \left(\frac{7 N_{gt}  \|O\| }{\epsilon} \right)^{d^{2k}N_{gt}}.  
\end{equation}
 
\end{proof}

\section{Proof of Proposition 1}\label{append:prop:cov-vqa-noise}
\begin{proof}[Proof of Proposition 1]
	In this proof, we first derive the covering number of VQA for the general noisy quantum channel $\mathcal{E}(\cdot)$, and then analyze the covering number of VQA when $\mathcal{E}(\cdot)$ specifies to the depolarization noise. 
	
	\textit{\underline{The general quantum channel $\mathcal{E}(\cdot)$.}} We follow the same routine as the proof of Theorem 1 to acquire the upper bound of $\mathcal{N}(\widetilde{\mathcal{H}},  \epsilon, |\cdot |)$. Namely, supported by Lemma \ref{lem:cov-num-two-space}, once we establish the relation between $d_2(\Tr(O\mathcal{E}(\hat{U}\rho  \hat{U}^{\dagger})), \Tr(O\mathcal{E}(\hat{U}_{\epsilon}\rho  \hat{U}_{\epsilon}^{\dagger})) )$, i.e., 
	\begin{equation}
		d_2\left(\Tr\left(O\mathcal{E}\left(\hat{U}\rho  \hat{U}^{\dagger}\right)\right), \Tr\left(O\mathcal{E}\left(\hat{U}_{\epsilon}\rho  \hat{U}_{\epsilon}^{\dagger}\right)\right) \right) = \left|\Tr\left(O\mathcal{E}\left(\hat{U}_{\epsilon}\rho  \hat{U}_{\epsilon}^{\dagger}\right) -  O \mathcal{E}\left(\hat{U} \rho \hat{U}^{\dagger}\right) \right) \right| ,
	\end{equation}
	and $d_1(\hat{U}_{\epsilon}\rho\hat{U}_{\epsilon}^{\dagger}, \hat{U}\rho\hat{U}^{\dagger})$,
	 the covering number $ \mathcal{N}(\mathcal{H}_{circ}, \epsilon, \|\cdot\|) $ in Lemma \ref{lem:thm-cov-num-opt-set} can be utilized to infer the upper bound of 	$\mathcal{N}(\widetilde{\mathcal{H}},  \epsilon, |\cdot |)$. 
	 
Under the above observation, we now derive the term $K$ such that $d_2(\Tr(O\mathcal{E}(\hat{U}\rho  \hat{U}^{\dagger})), \Tr(O\mathcal{E}(\hat{U}_{\epsilon}\rho  \hat{U}_{\epsilon}^{\dagger})) ) \leq K d_1(\hat{U}_{\epsilon}\rho\hat{U}_{\epsilon}^{\dagger}, \hat{U}\rho\hat{U}^{\dagger})$. In particular, we have
\begin{eqnarray}\label{eqn:proof-prop1-1}
	&& d_2\left(\Tr\left(O\mathcal{E}\left(\hat{U}\rho  \hat{U}^{\dagger}\right)\right), \Tr\left(O\mathcal{E}\left(\hat{U}_{\epsilon}\rho  \hat{U}_{\epsilon}^{\dagger}\right)\right) \right) \nonumber\\
	= && \left|\Tr\left(O \left(\mathcal{E}\left(\hat{U}_{\epsilon}\rho  \hat{U}_{\epsilon}^{\dagger}\right) - \mathcal{E}\left(\hat{U} \rho  \hat{U}^{\dagger}\right)   \right) \right) \right|  \nonumber\\
	\leq &&   \|O\| \Tr\left(  \mathcal{E}\left(\hat{U}_{\epsilon}\rho  \hat{U}_{\epsilon}^{\dagger}\right) -    \mathcal{E}\left(\hat{U} \rho \hat{U}^{\dagger}\right)   \right)    \nonumber\\
	\leq &&  \|O\|  \Tr\left(   \hat{U}_{\epsilon}\rho  \hat{U}_{\epsilon}^{\dagger}  -     \hat{U} \rho \hat{U}^{\dagger}  \right)  \nonumber\\ 
	\leq &&   2\|O\|    \left\| \hat{U}_{\epsilon}\rho  \hat{U}_{\epsilon}^{\dagger}  -     \hat{U} \rho \hat{U}^{\dagger}\right\|,  
\end{eqnarray}
where the first inequality uses the Cauchy-Schwartz inequality, the second inequality employs the contractive property of quantum channels (Theorem 9.2, \cite{nielsen2010quantum}), the last inequality comes from the fact that $\hat{U}_{\epsilon}\rho  \hat{U}_{\epsilon}^{\dagger}$ and $U\rho  U^{\dagger}$ are two rank-1 states (i.e., this implies that the rank of $\hat{U}_{\epsilon}\rho  \hat{U}_{\epsilon}^{\dagger}  -     \hat{U} \rho \hat{U}^{\dagger}$ is at most 2)  and $\Tr(\cdot)\leq rank(\cdot)\|\cdot\|$. 

With setting the operator $O$ in $d_1$ as $\rho$, we obtain 
\begin{equation}
	d_2\left(\Tr\left(O\mathcal{E}\left(\hat{U}\rho  \hat{U}^{\dagger}\right)\right), \Tr\left(O\mathcal{E}\left(\hat{U}_{\epsilon}\rho  \hat{U}_{\epsilon}^{\dagger}\right)\right) \right)    \leq  \|O\| 2 d_1(\hat{U}_{\epsilon}\rho \hat{U}_{\epsilon}^{\dagger}, U\rho U^{\dagger}),
\end{equation}
which indicates that the term $K$ in Eq.~(\ref{eqn:lem-cov-num-two-space-2}) is
\begin{equation}
	K = 2\|O\|. 
\end{equation} 

 Supporting by Lemma \ref{lem:thm-cov-num-opt-set}, the covering number of VQA under the noisy setting is upper bounded by
 \begin{equation}
 \mathcal{N}(\widetilde{\mathcal{H}},  \epsilon, |\cdot |) \leq 2\|O\|  \left(\frac{7 N_{gt}  \|\rho\| }{\epsilon} \right)^{d^{2k}N_{gt}}= 2\|O\|  \left(\frac{7 N_{gt}  }{\epsilon} \right)^{d^{2k}N_{gt}},
 \end{equation} 
 where the equality exploits the spectral property of the quantum state.

\textit{\underline{The local depolarization  channel $\mathcal{E}_p(\cdot)$.}}  We next consider the covering number of VQA when the noisy quantum channel is simulated by the local depolarization noise, i.e., the depolarization channel $\mathcal{E}_p(\cdot)$ is applied to each quantum gate in $\hat{U}(\bm{\theta})$. Following the explicit form of the  depolarization channel, the distance $d_2(\Tr(O\mathcal{E}_p(\hat{U}\rho  \hat{U}^{\dagger})), \Tr(O\mathcal{E}_p(\hat{U}_{\epsilon}\rho  \hat{U}_{\epsilon}^{\dagger})) )$ and distance $d_1(\hat{U}_{\epsilon}\rho\hat{U}_{\epsilon}^{\dagger}, \hat{U}\rho\hat{U}^{\dagger})$ satisfies
\begin{eqnarray}
	&& d_2\left(\Tr\left(O\mathcal{E}_p\left(\hat{U}\rho  \hat{U}^{\dagger}\right)\right), \Tr\left(O\mathcal{E}_p\left(\hat{U}_{\epsilon}\rho  \hat{U}_{\epsilon}^{\dagger}\right)\right) \right)  \nonumber\\
	= &&  \left|\Tr\left(O\mathcal{E}_{p}\left(\hat{U}_{\epsilon}\rho  \hat{U}_{\epsilon}^{\dagger}\right) -  O \mathcal{E}_{p}\left(\hat{U} \rho \hat{U}^{\dagger}\right) \right) \right| \nonumber\\
	= && (1-p)^{N_g}\left|\Tr\left(O \left(\hat{U}_{\epsilon}\rho  \hat{U}_{\epsilon}^{\dagger}\right) -  O  \left(\hat{U} \rho \hat{U}^{\dagger}\right) \right) \right| \nonumber\\
	\leq && (1-p)^{N_g} \left \|\hat{U}_{\epsilon}^{\dagger} O  \hat{U}_{\epsilon}    -  \hat{U}^{\dagger} O \hat{U}  \right\| \Tr(\rho) \nonumber\\
	= && (1-p)^{N_g} \left \|\hat{U}_{\epsilon}^{\dagger} O  \hat{U}_{\epsilon}    -  \hat{U}^{\dagger} O \hat{U} \right\| \nonumber\\
	= && (1-p)^{N_g} d_1(\hat{U}_{\epsilon}^{\dagger}O\hat{U}_{\epsilon}, \hat{U}^{\dagger}O\hat{U}),
\end{eqnarray}
where the second equality comes from the property of the local depolarization noise given in \cite[Lemma 5]{du2020learnability}, i.e.,
\begin{equation}
\mathcal{E}_{p}\left(\hat{U} \rho \hat{U}^{\dagger}\right) =  \mathcal{E}_p(u_{N_g}(\bm{\theta})...{u}_2(\bm{\theta}) \mathcal{E}_p({u}_1(\bm{\theta})\rho {u}_1(\bm{\theta})^{\dagger}){u}_2(\bm{\theta})^{\dagger}...u_{N_g}(\bm{\theta})^{\dagger}) = (1-p)^{N_g} (\hat{U}(\bm{\theta})  \rho \hat{U}(\bm{\theta})^{\dagger}) + (1- (1-p)^{N_g})\frac{\mathbb{I}}{N^d}.
\end{equation}
This result indicates that the term $K$ in Eq.~(\ref{eqn:lem-cov-num-two-space-1}) is
\begin{equation}
	K = (1-p)^{N_g}.
\end{equation} 

 Supporting by Lemma \ref{lem:thm-cov-num-opt-set}, the covering number of VQA under the depolarization noise is upper bounded by
 \begin{equation}
 	\mathcal{N}(\widetilde{\mathcal{H}},  \epsilon, |\cdot |) \leq (1-p)^{N_g}   \left(\frac{7 N_{gt}  \|O\| }{\epsilon} \right)^{d^{2k}N_{gt}}.   
 \end{equation}
\end{proof}

\section{Proof of Theorem 2}\label{append:proof-general-error-QNN}
\begin{lemma}[Theorem 1,  \cite{Kakade2008OnTC}] \label{lem:gene-rademachier}
Assume the loss $\ell$ is $L_1$-Lipschitz and upper bounded by $C_1$. With probability at least $1-\delta$ over a sample $\mathcal{S}$ of size $n$, every $h\in \HQNN$ satisfies 
\begin{equation}
\mathcal{R}(\mathcal{A}(S)) 	\leq  \hat{\mathcal{R}}_S(\mathcal{A}(S))  + 2L_1\Re(\HQNN)  +  3C_1\sqrt{\frac{\ln(2/\delta)}{2n}},
\end{equation}
where $\Re(\HQNN)$ is the empirical Rademacher complexity of the hypothesis space $\HQNN$ and $n$ is the sample size of $\mathcal{S}$. 
\end{lemma}

\begin{proof}[Proof of Theorem 2]
The result of Lemma \ref{lem:gene-rademachier} indicates that the precondition to infer the generalization error is deriving the upper bound of the Rademacher complexity $\Re(\HQNN)$. {The matematical expression of the Rademacher complexity is
\begin{equation}
	\Re(\mathcal{H}_{\HQNN}):=n^{-1}\mathbb{E}(\sup_{h\in \mathcal{H}}\sum_{i=1}^n \epsilon_i h(x_i,y_i)),
\end{equation}   
where the expectation is over the Rademacher random variables ($\epsilon_1,...,\epsilon_n$), which are i.i.d with $\Pr[\epsilon_1=1]=\Pr[\epsilon_1=-1]=1/2$.}  
To achieve this goal, we employ \revise{the Dudley entropy integral bound \cite{dudley1967sizes}} to connect Rademacher complexity with covering number, i.e., 
\begin{equation}\label{eqn:thm:generalization-QNN-1}
	\Re(\HQNN) \leq \inf_{\alpha>0}\left(4\alpha + \frac{12}{\sqrt{n}}\int_{\alpha}^{1}\sqrt{\ln \mathcal{N}((\HQNN)_{|S},  \epsilon, \|\cdot\|_2)} d \epsilon \right), 
\end{equation} 
where $(\HQNN)_{|S}$ denotes the set of vectors formed by the hypothesis with $n$ examples, i.e., $
	 \{  [h_{\AQNN(S)}(\bm{x}^{(i)}) ]_{i=1:n} \big| \bm{\theta}\in \Theta  \}$.
	 
We first establish the relation between the covering number of $ (\HQNN)_{|S}$ and $(\HQNN)_{|\bm{x}^{(i)}}$ to derive the upper bound of $\ln \mathcal{N}((\HQNN)_{|S},  \epsilon, \|\cdot\|_2)$. As with Lemma \ref{lem:thm-cov-num-opt-set}, denote a fixed $(\epsilon/\sqrt{n})$-covering $\mathcal{S}$ for the set $(\HQNN)_{|\bm{x}^{(i)}}$. Then for any function $
	h_{\AQNN(S)}(\cdot)\in \HQNN$ in Eq.~(6), we can always find a $h'_{\AQNN(S)}(\bm{x}^{(i)})\in \mathcal{S}$ such that $\forall i\in[n]$, $|h_{\AQNN(S)}(\bm{x}^{(i)})- h'_{\AQNN(S)}(\bm{x}^{(i)})|\leq \epsilon/\sqrt{n}$, and the discrepancy $\|[h_{\AQNN(S)}(\bm{x}^{(i)}) ]_{i=1:n} - [h'_{\AQNN(S)}(\bm{x}^{(i)}) ]_{i=1:n}\|_2$ satisfies
 \begin{eqnarray}\label{eqn:append-proof-lemma2-1}
  && \left\|[h_{\AQNN(S)}(\bm{x}^{(i)}) ]_{i=1:n} - [h'_{\AQNN(S)}(\bm{x}^{(i)}) ]_{i=1:n}\right\|_2 \nonumber\\
  = && \sqrt{ \sum_{i=1}^n |h_{\AQNN(S)}(\bm{x}^{(i)})   -  h'_{\AQNN(S)}(\bm{x}^{(i)}) |^2} \nonumber\\
  \leq &&  \epsilon.
 \end{eqnarray} 
Therefore, by Definition 1, we know that $\mathcal{S}$ is a $\epsilon$-covering set for $(\HQNN)_{|S}$. This result gives  
\begin{equation}\label{eqn:append-thm2-1}
	\ln \left(\mathcal{N}\left((\HQNN)_{|S},  \epsilon, \|\cdot\|_2 \right) \right) 
\leq    \ln \left(\mathcal{N}\left((\HQNN)_{|\bm{x}^{(i)}},  \frac{\epsilon}{\sqrt{n}}, |\cdot|  \right) \right).
\end{equation}

The right hand-side in Eq.~(\ref{eqn:append-thm2-1}) can be further upper bounded as
\begin{eqnarray}
	&&  \ln \left(\mathcal{N}\left((\HQNN)_{|\bm{x}^{(i)}},  \frac{\epsilon}{\sqrt{n}}, |\cdot|  \right) \right)   \nonumber\\
\leq  &&  \ln \left( \left(\frac{7 \sqrt{n} N_{gt}  \|O\| }{\epsilon} \right)^{d^{2k}N_{gt}} \right)  \nonumber\\
=  && d^{2k}N_{gt} \ln \left(  \frac{7 \sqrt{n}N_{gt}  \|O\|}{\epsilon} \right),
\end{eqnarray}
where the first inequality can be easily derived based on the proof of Lemma \ref{lem:thm-cov-num-opt-set} and the second inequality uses the result of Theorem 1. To this end, the integration term in Eq.~(\ref{eqn:thm:generalization-QNN-1}) follows
\begin{eqnarray}\label{eqn:thm:generalization-QNN-4}
	&& \frac{12}{\sqrt{n}}\int_{\alpha}^{1}\sqrt{\ln \mathcal{N}((\HQNN)_{|S},  \epsilon, \|\cdot\|_2)} d\epsilon \nonumber\\
	\leq &&	\frac{12}{\sqrt{n}}\int_{\alpha}^{1}   \sqrt{d^{2k}N_{gt} \ln \left(  \frac{7\sqrt{n} N_{gt}  \|O\|}{\epsilon} \right)} d\epsilon \nonumber\\
	\leq && \frac{12}{\sqrt{n}}\int_{\alpha}^{1}   d^{k}\sqrt{N_{gt}} \ln \left(  \frac{7 \sqrt{n} N_{gt}  \|O\|}{\epsilon} \right)  d\epsilon \nonumber\\
	= && \frac{12}{\sqrt{n}} d^{k}\sqrt{N_{gt}}\epsilon\left(\ln\left(\frac{7 \sqrt{n} N_{gt}  \|O\|}{\epsilon}\right) + 1 \right)\Big|_{\epsilon=\alpha}^{1} \nonumber\\
	= && \frac{12}{\sqrt{n}} d^{k}\sqrt{N_{gt}} \left(\ln\left(7 \sqrt{n} N_{gt}  \|O\| \right)+1\right) -  \frac{12}{\sqrt{n}} d^{k}\sqrt{N_{gt}}\alpha\left(\ln\left(\frac{7 \sqrt{n} N_{gt}  \|O\|}{\alpha}\right)+1\right)
\end{eqnarray}
  where the first inequality employs the upper bound of the covering number of $\HQNN$ in Theorem 2, and the second inequality uses the monotony of integral.  
  
 For simplicity, we set $\alpha=1/\sqrt{n}$ in Eq.~(\ref{eqn:thm:generalization-QNN-1}) and then the Rademacher complexity $\Re(\HQNN)$ is upper bounded by
  \begin{equation}\label{eqn:thm-gene-qnn-2}
  	\Re(\HQNN) \leq \frac{4}{\sqrt{n}} +  \frac{12}{\sqrt{n}} d^{k}\sqrt{N_{gt}} \left(\ln\left(7 \sqrt{n} N_{gt} \|O\| \right)+1\right).
  \end{equation}

In conjunction Lemma \ref{lem:gene-rademachier} with Eq.~(\ref{eqn:thm-gene-qnn-2}), with probability $1-\delta$, the generalization bound of QNN yields
  \begin{equation}
  	 \mathcal{R}(\mathcal{A}(S))-\hat{\mathcal{R}}_S(\mathcal{A}(S))  \leq        \frac{8L_1}{\sqrt{n}} +  \frac{24L_1}{\sqrt{n}} d^{k}\sqrt{N_{gt}} \left(\ln\left(7 \sqrt{n} N_{gt} \|O\| \right)+1\right) +   3C_1\sqrt{\frac{\ln(1/\delta)}{2n}}.
  \end{equation}   
\end{proof}

\section{Expressivity of other advanced quantum neural networks}\label{append:express-QNN-other}
To better understand how the covering number effects the expressivity of VQAs, in this section, we explicitly quantify the covering number of QNNs with several representative Ans\"atze, i.e., the hardware-efficient Ans\"atze, the tensor-network based Ans\"atze with the matrix product state structure, and the tensor-network based Ans\"atze with the tree structure. 

\textbf{Hardware-efficient Ans\"atze.} We first quantify the expressivity of QNN proposed by \cite{havlivcek2019supervised}, where $\hat{U}(\bm{\theta})$ is implemented by the hardware-efficient Ans\"atze, under the both the ideal and NISQ settings. An $N$-qubits hardware-efficient Ansatz is composed of $L$ layers, i.e., $U(\bm{\theta})=\prod_{l=1}^LU(\bm{\theta}^l)$ with $L\sim poly(N)$. For all layers, the arrangement of quantum gates in $U(\bm{\theta}^{l})$ is identical, which generally consists of parameterized single-qubit gates and fixed two-qubit gates. Moreover,  each qubit is operated with at least one parameterized single-qubit gate, and two qubits gates within the layer can adaptively connect two qubits depending on the qubits connectivity of the employed quantum hardware. An example of the $7$-qubits hardware-efficient Ansatz is illustrated in the left panel of Fig.~\ref{fig:QNN-hardware}. The parameterized single-qubit gate $U$ can be realized by the rotational qubit gates, e.g., $U\in \{R_X(\theta), R_Y(\theta), R_Z(\theta)\}$ or $U=R_Z(\beta)R_Y(\gamma)R_Z(\nu)$ with $\theta, \gamma, \beta, \nu\in[0, 2\pi)$.  The topology of two-qubit gates, i.e., CNOT gates, aims to adapt to the chain-like connectivity restriction. 

The hardware-efficient Ansatz considered here is the most general case.  Specifically, the single-qubit gate $U$ contains three trainable parameters and the number of two-qubit gates in each layer is set as $N$. Under this setting, the total number of quantum gates in $U(\bm{\theta})=\prod_{l=1}^LU(\bm{\theta}^l)$ is
\begin{equation}
N_g=	L(3N+N) = 4LN.
\end{equation}
Based on the above settings, we achieve the expressivity of QNN with the hardware-efficient Ans\"atze, supported by Theorem 1 and Proposition 1. 
\begin{corollary}\label{coro:QNN} 
	Under the ideal setting, the covering number of QNN with the hardware-efficient Ansatz is upper bounded by $\mathcal{N}(\HQNN,  \epsilon, |\cdot|) \leq  (\frac{21 NL \|O\|}{\epsilon} )^{6NL}$. When the system noise is considered and simulated by the depolarization channel, the corresponding covering number is upper bounded by $\mathcal{N}(\widetilde{\HQNN},  \epsilon, |\cdot|) \leq  (1-p)^{4NL}   (\frac{21 NL \|O\|}{\epsilon} )^{6NL}$.
\end{corollary}   
 
  \begin{figure}
 	\centering
 	\includegraphics[width=0.98\textwidth]{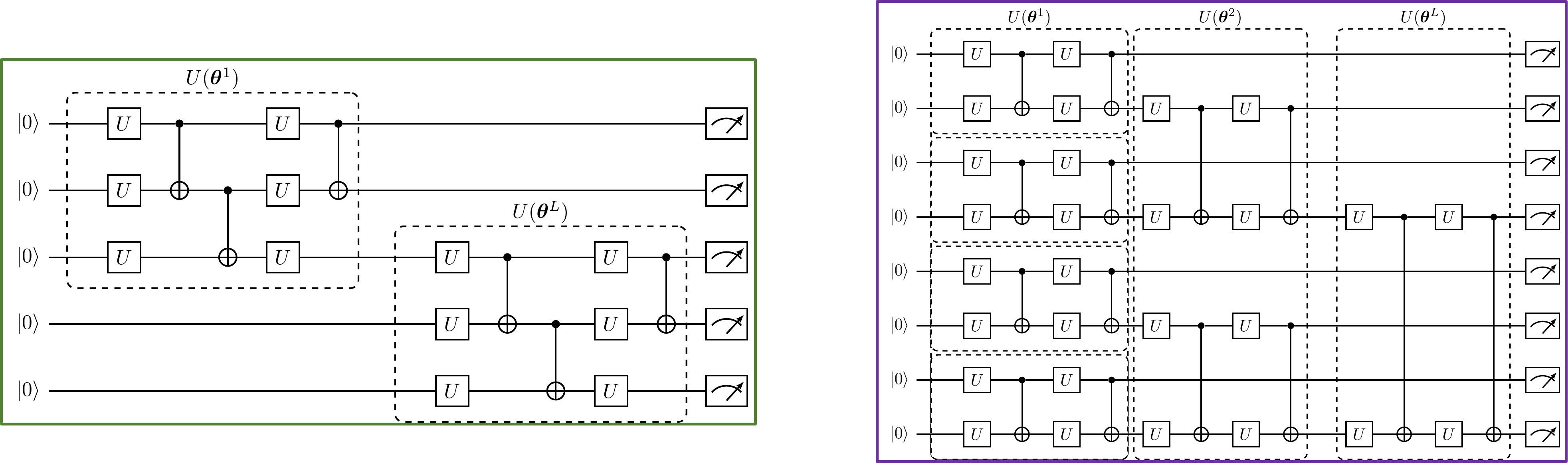}
 	\caption{\small{\textbf{Illustration of two tensor-network based Ans\"atze used in QNNs.} The left panel presents the  tensor network based Ans\"atze with the matrix product state structure (highlighted by the green box), and the right panel refers to the tensor network based Ans\"atze with the tree structure (highlighted by the purple box). }}
 	\label{fig:Ansatz-QNN}
 \end{figure} 
 
\textbf{Tensor-network based Ans\"atze with the matrix product state structure.} Another type of Ans\"atze  inherits the tensor-network structures, i.e., matrix product states and tree tensor network \cite{huggins2019towards}. The left panel in Fig.~\ref{fig:Ansatz-QNN} illustrates the tensor-network based Ans\"atze with the matrix product state structure. Mathematically, for an $N$-qubit quantum circuit, the corresponding Ans\"atze yields
\begin{equation}
	\hat{U}(\bm{\theta})= \prod_{l=1}^{L}\left(\mathbb{I}_{2^{(M_1-1) * (l-1)} }\otimes U(\bm{\theta}^l)\otimes \mathbb{I}_{2^{N - 1 - (M_1-1) * l}} \right),
\end{equation}
where $U(\bm{\theta}^l)$ is applied to $M_1$ qubits for $\forall l \in [L]$ with $2\leq M_1 < N$. The topology as shown in Fig.~\ref{fig:Ansatz-QNN} indicates that the maximum circuit depth of the tensor-network based Ans\"atze with the matrix product state structure is $L=\lceil N/(M_1-1) \rceil$. Suppose that the total number of single-qubit and two-qubit quantum gates in $U(\bm{\theta}^l)$ is $3M_1$ and $M_1$ respectively, we have 
\begin{equation}
	N_g = 4M_1 \lceil N/(M_1-1) \rceil \leq 4(N+M_1+N/M_1) \leq 4(N+2\sqrt{N}).  
\end{equation} 
Based on the above settings, we achieve the expressivity of QNN with tensor-network based Ans\"atze with the matrix product state structure, supported by Theorem 1 and Proposition 1. 
\begin{corollary}\label{coro:QNN-MPS} 
	Under the ideal setting, the covering number of QNN with tensor-network based Ans\"atze with the matrix product state structure is upper bounded by $\mathcal{N}(\HQNN,  \epsilon, |\cdot|) \leq  \left(\frac{21(N+2\sqrt{N}) \|O\|}{\epsilon} \right)^{6(N+2\sqrt{N})}$. When the system noise is considered and simulated by the depolarization channel, the corresponding covering number is upper bounded by $\mathcal{N}(\widetilde{\HQNN},  \epsilon, |\cdot|) \leq  (1-p)^{4(N+2\sqrt{N})}   \left(\frac{21(N+2\sqrt{N}) \|O\|}{\epsilon} \right)^{6(N+2\sqrt{N})}$.
\end{corollary}

\textbf{Tensor-network based Ans\"atze with the tree structure.} The right panel in Fig.~\ref{fig:Ansatz-QNN} illustrates the tensor-network based Ans\"atze with tree structure.  Intuitively, the involved number of quantum gates is exponentially decreased in terms of $l\in[L]$. Suppose that the local unitary, as highlighted by the dotted box in the right panel of Fig.~\ref{fig:Ansatz-QNN} with $l=1$, contains six single-qubit gates (i.e., each qubit is operated with $R_Z(\beta)R_Y(\gamma)R_Z(\nu)$) and one two-qubit gates. Then for an $N$-qubit quantum circuit, the total number of quantum gates in $\hat{U}$ is 
\begin{equation}
	N_g = 7 \lceil N/2 \rceil + 7 \lceil N/4 \rceil + ... + 7 *2 \leq 7 N.
\end{equation} 
Based on the above settings, we achieve the expressivity of QNN with tensor-network based Ans\"atze with the matrix product state structure, supported by Theorem 1 and Proposition 1. 
\begin{corollary}\label{coro:QNN-tree} 
	Under the ideal setting, the covering number of QNN with tensor-network based Ans\"atze with the matrix product state structure is upper bounded by $\mathcal{N}(\HQNN,  \epsilon, |\cdot|) \leq  \left(\frac{73.5N \|O\|}{\epsilon} \right)^{10.5N}$. When the system noise is considered and simulated by the depolarization channel, the corresponding covering number is upper bounded by $\mathcal{N}(\widetilde{\HQNN},  \epsilon, |\cdot|) \leq  (1-p)^{7N}   \left(\frac{73.5N \|O\|  }{\epsilon} \right)^{10.5N}$.
\end{corollary}

 \section{Numerical simulation details of QNN}\label{append:QNN}

\textbf{Implementation.} The implementation of QNN employed in the numerical simulations is shown in the right panel of Fig.~\ref{fig:QNN-hardware}. In particular, the qubit encoding method \cite{larose2020robust} is exploited to load classical data into quantum forms. The explicit form of the encoding circuit is
 \begin{equation}\label{eqn:append-QNN-encoding}
 	U_E(\bm{x}^{(i)}) = U_{\text{Eng}} \left(\bigotimes_{j=1}^7 R_Y(\bm{x}^{(i)}_j)\right) U_{\text{Eng}} \left(\bigotimes_{j=1}^7 R_Y(\bm{x}^{(i)}_j)\right),  
 \end{equation}      
 where the unitary $U_{\text{Eng}}$ is formed by CNOT gates as shown in Fig.~\ref{fig:QNN-hardware}. The parameterized single-qubit qubit used in the Ans\"atze yields $U(\bm{\theta}^{l}_j)=R_Z(\beta)R_Y(\gamma)R_Z(\nu)$ for $\forall j\in[N]$ and $\forall l\in[L]$, where $\beta, \gamma, \nu\in[0, 2\pi)$ are independent trainable parameters.

\begin{figure}
	\centering
\includegraphics[width=0.82\textwidth]{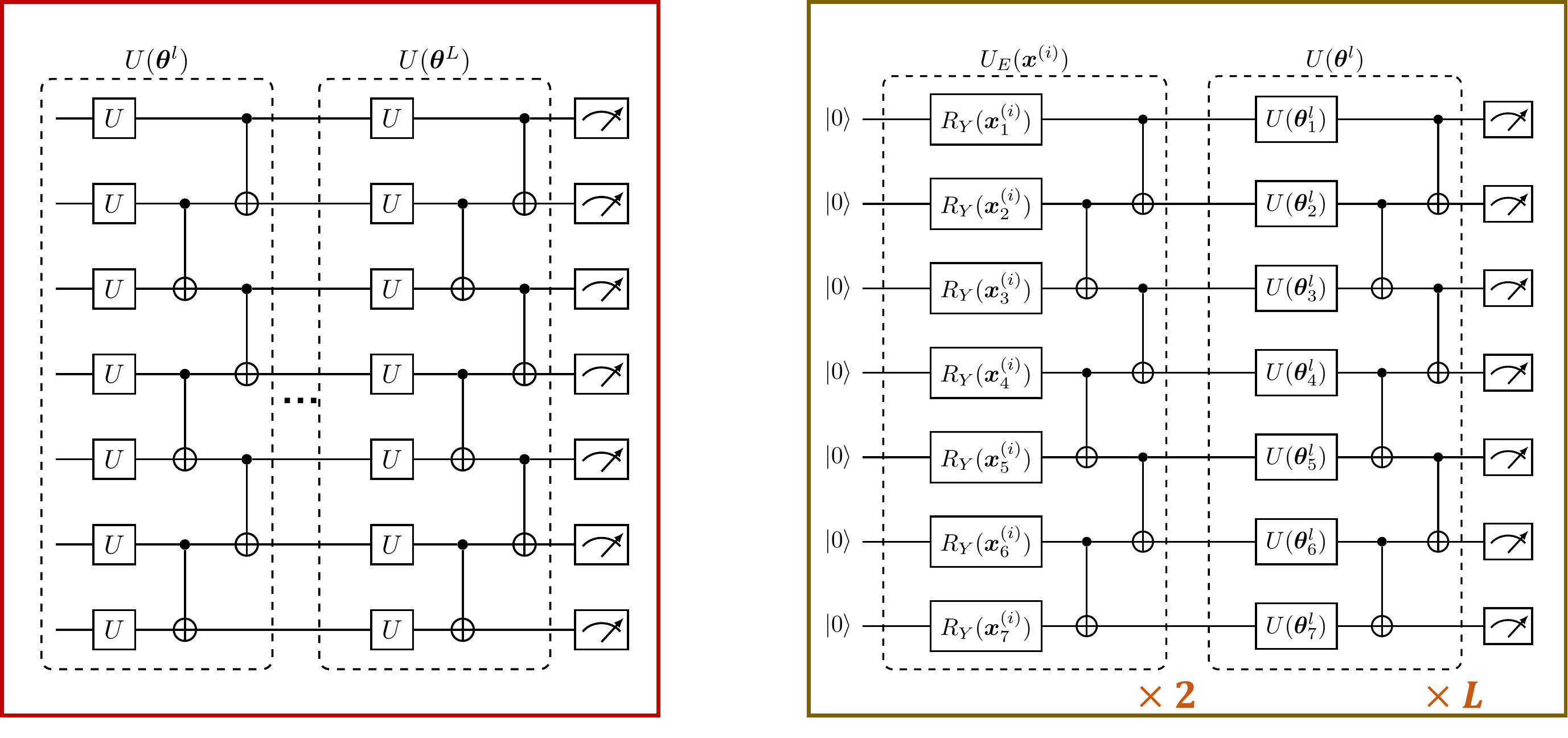}
	\caption{\small{\textbf{QNN with the hardware-efficient Ans\"atze}. The left panel depicts the $7$-qubit hardware-efficient Ansatz (highlighted by the red box). The right panel illustrates the implementation of QNN used in the numerical simulation (highlighted by the yellow box). }}
	\label{fig:QNN-hardware}
\end{figure}
\textbf{Data construction.} The construction of the synthetic dataset $S=\{\bm{x}^{(i)},y^{(i)}\}_{i=1}^{n}$  imitates the studies \cite{Du_2021_grover,havlivcek2019supervised}. Specifically, for each example, the feature dimension of $\bm{x}^{(i)}$ is set as $7$, i.e.,  $\bm{x}^{(i)}= [\omega_1^{(i)},\omega_2^{(i)}, \omega_3^{(i)}, \omega_4^{(i)}, \omega_5^{(i)}, \omega_6^{(i)}, \omega_7^{(i)}]^{\top} \in[0, 2\pi)^7$, and the label $y^{(i)}\in\{0, 1\}$ is binary. The assignment of the label $y^{(i)}$ is accomplished as follows. Define $V\in \mathcal{SU}(2^7)$ as a fixed unitary operator, $O=\mathbb{I}_{2^6}\otimes \ket{0}\bra{0}$ as the measurement operator, and the gap threshold $\Delta$ is set as $0.2$. The label of $\bm{x}^{(i)}$ is assigned as `1' if 
\begin{equation}
	\langle 0^{\otimes 7} |U_E(\bm{x}^{(i)})^{\dagger}V^{\dagger}O VU_E(\bm{x}^{(i)})|0^{\otimes 7} \rangle \geq 0.5 + \Delta;
\end{equation}
The label of $\bm{x}^{(i)}$ is assigned as `0' if 
\begin{equation}
	\langle 0^{\otimes 7} |U_E(\bm{x}^{(i)})^{\dagger}V^{\dagger}O VU_E(\bm{x}^{(i)})|0^{\otimes 7} \leq 0.5 -\Delta.
\end{equation}
We note that $V$ is realized by the Ans\"atze   $U(\bm{\theta}^*)=\prod_{l=1}^2U(\bm{\theta}^{*l})$ shown in Fig.~\ref{fig:QNN-hardware}, where the corresponding parameters $\bm{\theta}^*$ are sampled with the random seed `1'. This setting ensures the target concept $V$ is always covered by the hypothesis space $\HQNN$ once $L\geq 2$.  

Based on the construction rule in the above two equations, we collect the dataset $S$ with $n=400$, where the positive and negative examples are equally distributed. We illustrate some examples of $S$ in the left panel of Fig.~\ref{fig:noisy-QNN}. Given access to $S$, we split the dataset into the training datasets with size $60$ and the test dataset with $340$.  

\textbf{Hyper-parameters setting.}  The hyper-parameters setting used in our experiment is as follows.  At each epoch, we shuffle the training set in $S$.  An epoch means that an entire dataset is passed forward through the quantum learning model, e.g., when the dataset contains $1000$ training examples, and only two examples are fed into the quantum learning model each time, then it will take $500$ iterations to complete $1$ epoch. The learning rate is set as $\eta = 0.2$. The batch gradient descent method is adopted to be the optimizer with batch size equal to $4$. 

\begin{figure*}
	\centering
\includegraphics[width=0.98\textwidth]{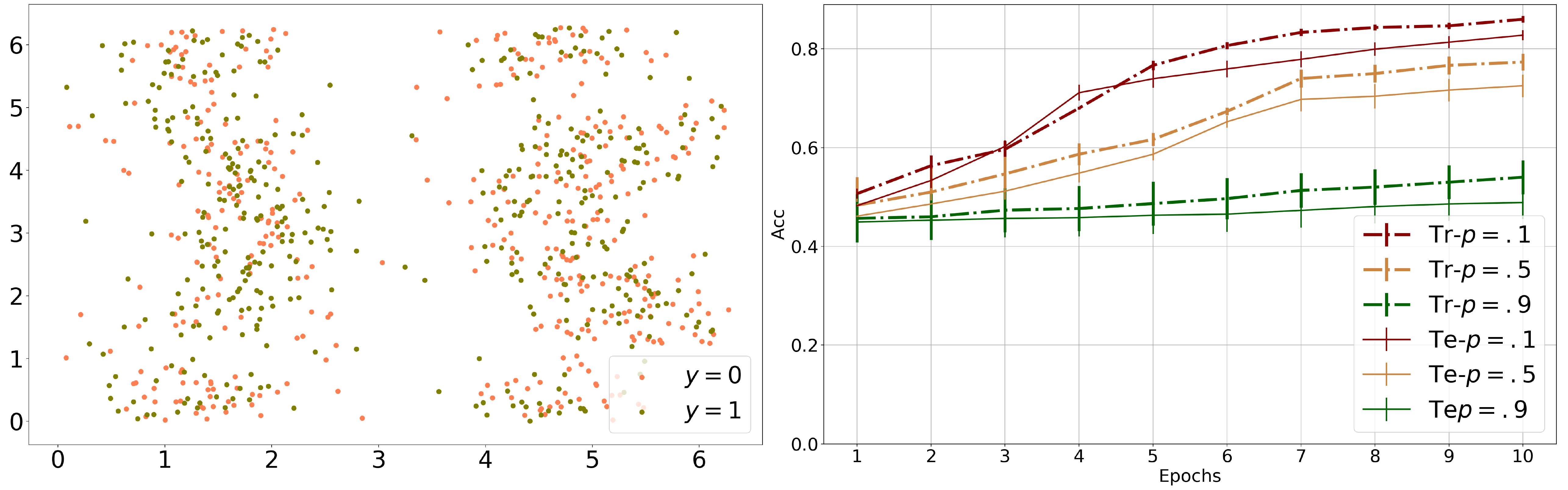}
\caption{\small{\textbf{The synthetic dataset and simulation results of noisy QNN.} The left plot illustrates the first two dimensions of the data points, where the green dots (pink dots) correspond to the data with label `1' (`0'). The right plot exhibits the learning performance of QNN when the depolarization noise is considered. The label `Tr-$p=a$' refers to the training accuracy of QNN when the depolarization rate is set as $p=a$. Similarly, the label `Te-$p=a$' refers to the test accuracy of QNN when the depolarization rate is set as $p=a$. }}
\label{fig:noisy-QNN}
\end{figure*}

\textbf{The performance of noisy QNNs.} Here we apply noisy QNN to learn the synthetic data $S$ introduced above to validate the correctness of Proposition 1. In particular, all settings, i.e., the employed Ans\"atze, the optimizer, and the hyper-parameters,    are identical to the noiseless case, except that the employed quantum circuit is interacted with the depolarization noise. With the aim of understanding how the depolarization rate $p$ shrinks the expressivity of $\HQNN$, we set the layer number of the hardware-efficient Ansatz as $L=2$ and the depolarization rate as $p\in\{0.1, 0.5, 0.9\}$. We repeat each setting with $5$ times to collect the statistical results.        

The simulation results are presented in the right panel of Fig.~\ref{fig:noisy-QNN}. Recall that  the training performance of the noiseless QNN with $L=2$ is above $85\%$ at the $10$-th epoch, as shown in Fig.~3. Meanwhile, the construction rule of $S$ indicates that the target concept is contained in $U(\bm{\theta} )=\prod_{l=1}^2U(\bm{\theta}^{ l})$. However, the   results in Fig.~\ref{fig:noisy-QNN} reflect that both the training and test accuracies continuously degrade in terms of the increased $p$. When $p=0.9$, the learning performance is around $50\%$, which is no better than the random guess. These observations accord with Proposition 1 such that an increased depolarization rate suppresses the expressivity of $\HQNN$ and excludes the target concept out of the hypothesis space, which leads to a poor learning performance. 

\section{Proof of Corollary 1}\label{append:VQE-UCCSD}
For completeness, let us first briefly introduce the unitary coupled-cluster Ans\"atze truncated up to single and double excitations (UCCSD)  before presenting the proof of Corollary 1. Please refer to Refs. \cite{cao2019quantum,romero2018strategies} for comprehensive explanations. UCCSD belongs to a special type of unitary coupled-cluster (UCC) operator,  which takes the form $e^{T-T^{\dagger}}$, where $T$ corresponds to excitation operators defined for the configuration interaction. Since the unitary $e^{T-T^{\dagger}}$ is difficult to implement on quantum computers,  an alternative Ans\"atze is truncating UCC up to single and double excitations, as so-called UCCSD, which can be used to accurately describe many molecular systems and is exact for systems with two electrons.   Mathematically, UCCSD estimates $T$ by $T_1+T_2$. The study \cite{cao2019quantum} has indicated that for both the Bravyi-Kitaev and the Jordan-Wigner transformations, the required number of quantum gates to implement UCCSD is upper bounded by $N_g\sim O(N^5)$. 

\begin{proof}[Proof of Corollary 1]
	The results of Corollary 1 can be immediately achieved by substituting $N_g\sim O(N^5)$ as explained above with Theorem 2 and Proposition 1. 
\end{proof}
  
\section{Numerical simulation details of VQE}\label{append:VQE}
\textbf{Implementation.} The implementation of VQEs employed in numerical simulations is shown in Fig.~\ref{fig:impt-VQE}. Based on the results in Theorem 1, we control the involved number of quantum gates to separate the expressivity of different Ans\"atze. In particular, the Ans\"atze as shown in left panel has a restricted expressivity, which only contains a single trainable quantum gate. The Ans\"atze as shown in middle panel has a modest expressivity, where $U(\bm{\theta}^{l}_j)=R_Z(\beta)R_Y(\gamma)R_Z(\nu)$ for $\forall j\in[N]$ and $\forall l\in[L]$. In other words, the total number of trainable quantum gates is $12$. Note that an Ans\"atze is sufficient to locate the minimum energy of $H$. The Ans\"atze as shown in middle panel has an overwhelming expressivity. Compared with the   Ans\"atze with the modest expressivity, the number of trainable quantum gates scales by four times. Such an over-parameterized model may suffer from the training hardness, caused by the barren plateaus phenomenon.

\textbf{The qubit Hamiltonians of the hydrogen molecule.} 
The Bravyi-Kitaev transformation \cite{bravyi2002fermionic} is used to attain the qubit Hamiltonian of the hydrogen molecule at each bond length. The mathematical form of the obtained qubit Hamiltonian yields
\begin{eqnarray}
	H = &&  f_0 \mathbb{I}  + f_1Z_0 + f_2Z_1 + f_3Z_2 +  f_1Z_0Z_1 + f_4Z_0Z_2 + f_5Z_1Z_3 + f_6X_0Z_1X_2 + f_6Y_0Z_1Y_2 + f_7Z_0Z_1Z_2 \nonumber\\
	&&  + f_4Z_0Z_2Z_3 + f_3Z_1Z_2Z_3 + f_6X_0Z_1X_2Z_3 + f_6Y_0Z_1Y_2Z_3 + f_7Z_0Z_1Z_2Z_3,
\end{eqnarray} 
where $\{X_i,Y_i,Z_i\}$ stands for applying the Pauli operators on the $i$-th qubit and the coefficients $\{f_j\}_{j=1}^7$  are determined by the bond length. In the numerical simulations, we use OpenFermion Library \cite{mcclean2020openfermion}  to load these coefficients.   
 
 \begin{figure*}
\centering
\includegraphics[width=0.98\textwidth]{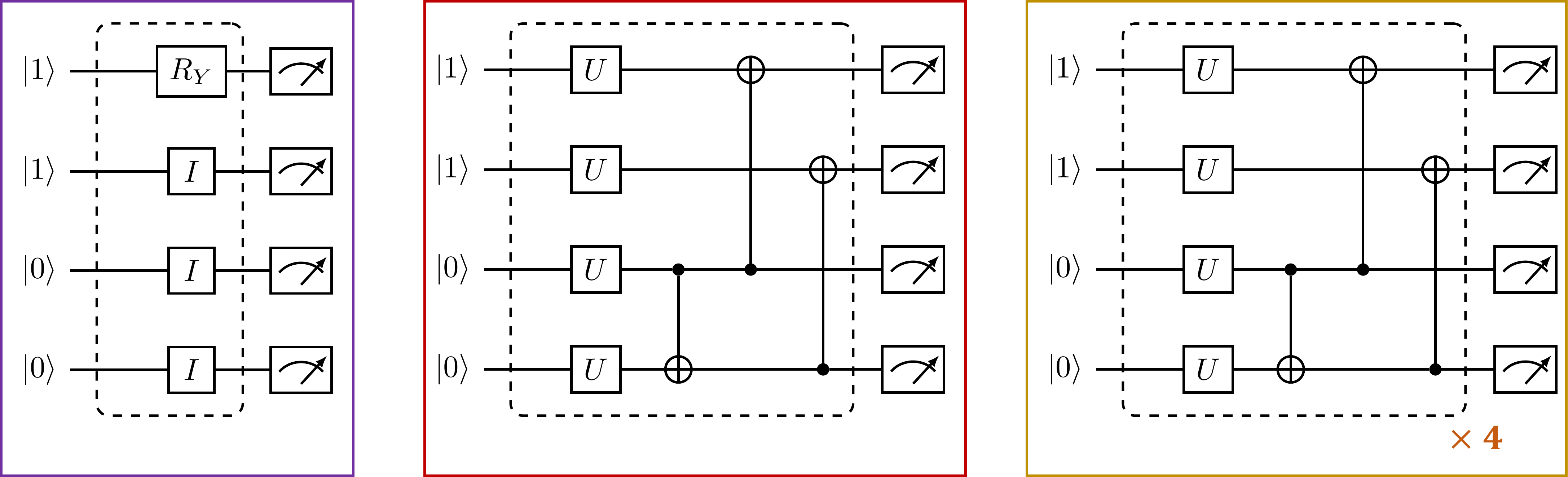}	
\caption{\small{\textbf{Implementation of  VQEs with different Ans\"atze}. The left, middle, and right panels depict the construction of VQE with restricted, modest, and overwhelming expressivity. The subscript `$\times 4$' in the right panel means repeating the circuit architecture in the dotted box with four times.  }}
\label{fig:impt-VQE}
\end{figure*}

\textbf{Hyper-parameters setting.}  The hyper-parameters setting related to the optimization of VQEs is as follows.  The total number of iteration is set as $300$. The tolerant error is set as $10^{-6}$. The gradient descent optimizer is adopted and the learning rate is set as $\eta = 0.4$. The random seed used to initialize trainable parameters is set as $0$. 

{
\textbf{The performance of VQEs with the Adam optimizer.} We conduct additional numerical simulations  to explore how the expressivity of Ans\"atze affects the performance of VQEs when the adaptive optimizer is adopted. More precisely, we aim to investigate whether VQE with the over-parameterized Ansatz can outperform VQE with the modest Ansatz when the adaptive optimizer is adopted. The appended numerical simulations mainly follow the setup introduced in the main text. In particular, the hardware-efficient VQEs with the different layer number $L$ are employed to estimate the ground state energy of the Hydrogen molecule with varied bond length ranging from $0.3 \mathring{A}$ to $2.1\mathring{A}$. Notably, we substitute the SGD optimizer with the Adam optimizer \cite{Kingma2014Adam} to update trainable parameters $\bm{\theta}$ of the Ansatz with the adpative learning rate. Mathematically, the updating rule of the Adam optimizer yieds
\begin{equation}
	\bm{\theta}^{(t+1)} =  \bm{\theta}^{(t)} - \eta^{(t+1)}\frac{a^{(t+1)}}{\sqrt{b^{(t+1)}}+\epsilon}, 
\end{equation} 
where $a^{(t+1)}=[\beta_1a^{(t)}+(1-\beta_1)\nabla \mathcal{L}(\bm{\theta}^{(t)})]/(1-\beta_1)$,  $b^{(t+1)}= [\beta_2b^{(t)}+(1-\beta_2)\nabla \mathcal{L}(\bm{\theta}^{(t)})^{\odot 2}] /(1-\beta_2)$, and $\eta^{(t+1)} = \eta^{(t)} \frac{\sqrt{(1-\beta_2)}}{(1-\beta_1)}$. The hyper-parameters settings are as follows. The maximum number of iterations are fixed to be $81$. \revise{The  optimization is criticized to be converged and the updating is stopped   when the energy difference between two iterations is lower than $10^{-6}$, i.e., $|\mathcal{L}(\bm{\theta}^{(t)}) - \mathcal{L}(\bm{\theta}^{(t-1)})|\leq 10^{-6}$.} The learning rate at the initial step is set as $\eta = 0.4$. The momentum parameters are set as $\beta_1=0.9$ and $\beta_2=0.99$. The tolerance parameter is set as $\epsilon=10^{-6}$. The random seed used to initialize trainable parameters is set as $0$. The number of layers of the hardware-efficient Ansatz, i.e., $U(\bm{\theta})=\prod_{l=1}^L U_l(\bm{\theta})$ exhibited in Fig.~\ref{fig:impt-VQE}, varies from $L=5$ to $L=20$. Each setting is repeated $5$ times to collect the statistical results. 
}

The achieved simulation results are depicted in Fig.~\ref{fig:VQE-adam}. Specifically, when the Adam optimizer is utilized to adaptively adjust the learning rate at each iteration, VQE with the modest Ansatz still attains a better performance than VQE with the overwhelming-expressivity Ansatz. Moreover, as shown in \revise{the inner plot of the left panel}, the performance of VQE continuously degrades    with respect to the increased number of layer number $L$. \revise{The right panel in Fig.~\ref{fig:VQE-adam} further exhibits the energy difference of VQE between the neighboring two iterations, i.e.,  $|\mathcal{L}(\bm{\theta}^{(t)}) - \mathcal{L}(\bm{\theta}^{(t-1)})|$, when the bond length is 0.3\AA. For the setting $L=5$, the optimization is converged when $t=79$ with $|\mathcal{L}(\bm{\theta}^{(79)}) - \mathcal{L}(\bm{\theta}^{(78)})|=9.4\times 10^{-7}$. For the setting $L=10, 15, 20$, the energy difference of VQE between the last two iterations yields $|\mathcal{L}(\bm{\theta}^{(79)}) - \mathcal{L}(\bm{\theta}^{(78)})|=1.9\times 10^{-4}, 8.9\times 10^{-4}, 3.9\times 10^{-2}$, respectively.} These observations indicate that our results still hold when the adaptive learning rate is considered. That is, too limited or too redundant expressivity of the employed Ans\"atze may prohibit the trainability of VQE.

\begin{figure*}
	\centering
\includegraphics[width=0.98\textwidth]{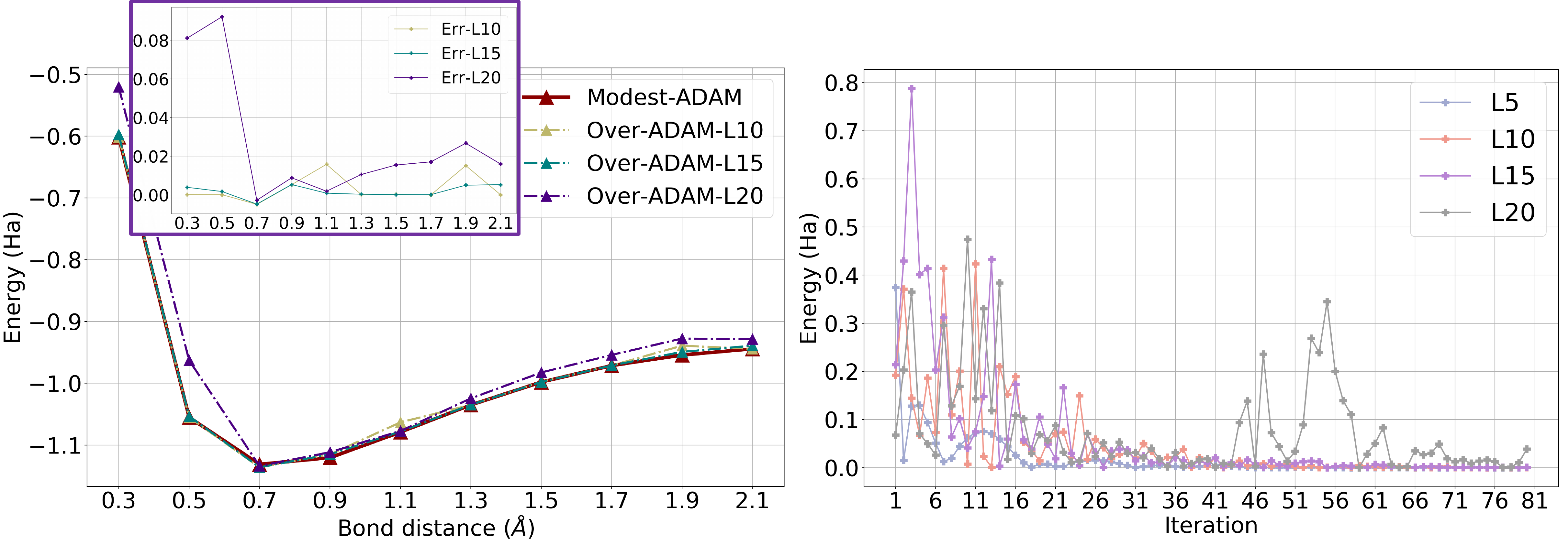}
\caption{\small{\textbf{Simulation results of VQE with the Adam optimizer.} The outer plot in the left panel illustrates the estimated energy of the hardware-efficient VQE with different number of layers $L$, i.e., $U(\bm{\theta})=\prod_{l=1}^L U_l(\bm{\theta})$. The label `Modest-Adam' refers to the estimated energy of VQEs with the modest Ansatz introduced in the main text and the Adam optimizer. The labels `over-Adam-L-$b$' refer to the estimated energy of VQEs when the employed Ans\"atze has the layer number $L=b$ and the Adam optimizer is employed.  The inner plot of the left panel shows the energy gap of VQEs in `Over-L-$b$' and `Modest' cases. \revise{`Ha' (Hartrees) and `\r{A}' (Angstroms) refer to the units for energy  and the bond lengths}. \revise{The right panel depicts the energy difference of VQE between the neighboring two iterations with the setting $L=5, 10, 15, 20$.} }}
\label{fig:VQE-adam}
\end{figure*}

\section{Implications of Theorem 1 and Theorem 2 from the practical perspective}\label{append:strm}

{In this section, we elucidate how our theoretical results, i.e., Theorems 1 and 2, contribute to improve the learning performance of  VQA-based models in practice. Concretely, the established theoretical results can be integrated with structural risk minimization \cite{mohri2018foundations} to enhance the learning performance of VQA-based models. Interestingly, the similar topic, i.e., the employment the expressivity of VQAs as guidance to enhance learning performance, has also been discussed in two very recent studies  \cite{caro2021encoding,gyurik2021structural}. 
}

{Before moving on to explain how our results contribute to the structural risk minimization of VQA-based learning models in practice, let us first recap the theory of structural risk minimization. As shown in Fig.~\ref{fig:SRM}, the learning performance of VQAs is determined by both the training error and the model's complexity term. Although the training error can be continuously suppressed by increasing the model's expressivity, the price to pay is increasing the complexity term, which may deteriorate its learning performance (e.g., the test accuracy for the unseen data). Overall,  increasing the expressivity of VQAs beyond a certain threshold is no longer contributing to the improvement of learning performance or could even lead to a degraded learning performance. With this regard, it is of great importance to develop an efficient method that seeks an Ansatz  with a `modest' level (i.e., the smile face in Fig.~\ref{fig:SRM}), which can well balance the tradeoff between the expressivity of the hypothesis space $\mathcal{H}$ and the performance of a learning model. In other words, the core of the VQA-based model design is controlling its expressivity at a `modest' level, envisioned by the statistical learning theory.  With this aim,  structural risk minimization is proposed as a concrete method that balances the trade-off between the expressivity and training error to attain the best possible learning performance. The mathematical expression of structural risk minimization can be formulated as
\begin{equation}\label{eqn:SRM}
	\min_{\bm{\theta}\in \bm{\Theta}, \Xi} \mathcal{\hat{R}}_S + g(S, \bm{\theta}, \Xi)~,
\end{equation}
where $\hat{\mathcal R}_S = \frac{1}{n} \sum_{i=1}^n \ell(h_{\mathcal{A}(S)}(\bm{x}^{(i)}), \bm{y}^{(i)})$ refers to the empirical risk (i.e., the training error term) defined in the main text and $g(\cdot)$ refers to the complexity term, which is controlled by the input problem $S$, the trainable parameters $\bm{\theta}$ and the architecture of learning models.  
}

\begin{figure*}
\centering
\includegraphics[width=0.5\textwidth]{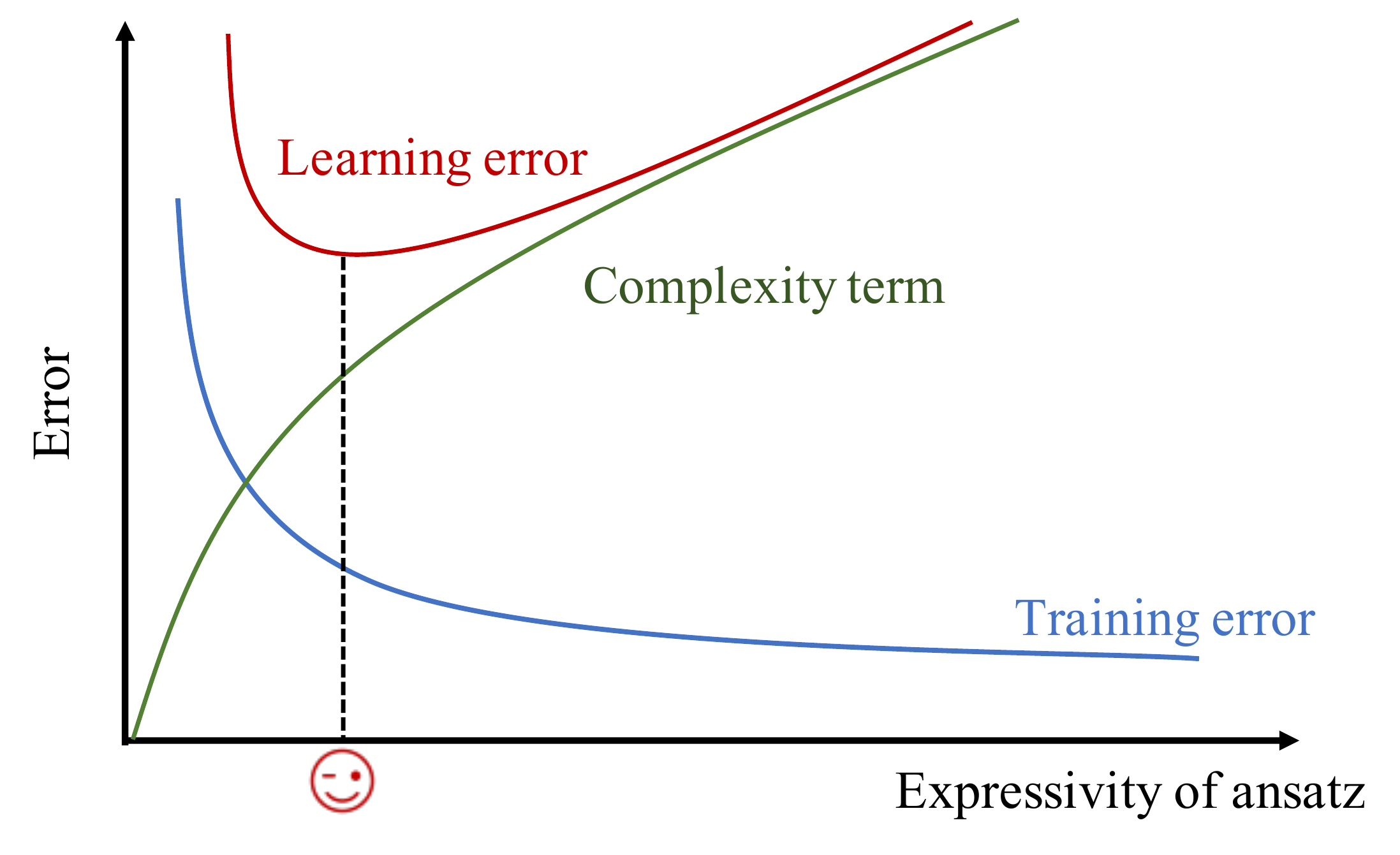}
\caption{\small{\textbf{ Illustration of the structural risk minimization adapted from \cite{mohri2018foundations}.}}}
\label{fig:SRM}
\end{figure*} 

{ 
We now detail three approaches that harness our theoretical results to engineer the complexity term $g(\cdot)$ and hence improve the learning performance of VQA-based learning models. 
\begin{enumerate}
	\item The first approach is setting $g(\cdot)$ as a regularizer with respect to the trainable parameters $\bm{\theta}$, e.g., $g(\cdot)=\lambda\|\bm{\theta}\|_2$ or $g(\cdot)=\lambda\|\bm{\theta}\|_0$, where $\lambda$ refers to a hyper-parameter. In this way, the optimal solution of the structural risk minimization in Eq.~(\ref{eqn:SRM}) implicitly controls the expressivity of the learning model by sparsifying the trainable parameters, ensured by our theoretical results that the expressivity of VQAs depends on the number of trainable parameters.
	\item The second approach is tailoring the spectral norm of the observable $O$ when it is trainable, supported by our theoretical results that the expressivity of VQAs depends on $\|O\|$. For instance, the complexity term can be set as $g(\cdot)=\lambda\|\bm{\theta}\|_2+\|O(\bm{\gamma})\|$, where $\bm{\gamma}$ denotes the parameters of the trainable observable.
	\item  The last approach is carefully designing the complexity term of $g(\cdot)$ that is determined by both the trainable parameters $\bm{\theta}$ and the quantum circuit architecture $\Xi$. The key motivation of updating the architecture of the quantum circuits is warranted by our theoretical results in Theorem 1, since the expressivity of VQAs depends on the adopted types of quantum gates (denoted by the term $k$). Following this routine, several studies have proposed different variable structure methods to build Ans\"atze \cite{bilkis2021semi,du2020quantum,grimsley2019adaptive,kuo2021quantum,ostaszewski2019quantum,tang2021qubit,zhang2021neural}. Conceptually, these proposals developed a set of heuristic rules that during the optimization,  the quantum gates in quantum circuits can either be added or deleted to find the optimal solution of the structural risk minimization in Eq.~(\ref{eqn:SRM}). 
\end{enumerate}
}

{
We remark that the first two approaches presented above have also been discussed in Refs.~\cite{caro2021encoding,gyurik2021structural}, where the analyzed expressivity of VQAs can be  applied to structural risk minimization to improve the performance of VQAs. However, the derived bounds in their results  are relatively loose and therefore fail to unveil the expressivity of VQAs is controlled by the types of quantum gates. With this regard, the achieved results in our study provide more concrete guidance (i.e., especially for the third approach) to implement structural risk minimization of VQA-based models in practice.
}

\section{The tightness of the derived upper bound in Theorem 1}

{The derivation of the upper and lower bounds with respect to the model's expressivity are at the heart of both classical and quantum machine learning. In the regime of deep learning, numerous studies \cite{pouyanfar2018survey,sun2019optimization} devote to     analyze the expressivity of deep neural networks (DNNs). Notably, most results focus on achieving a tighter upper bound for the expressivity of DNNs, while  only few studies attempt to analyze the corresponding lower bound. This phenomenon is caused by the fact that the derivation of the lower bound is  more difficult than the upper bound case. For instance, the seminal paper \cite{bartlett2017spectrally} only proved that the expressivity of DNNs is only lower bounded by the spectral norm of the input data and is independent with other parameters. The development of quantum learning theory also encounters the similar scenario. To our best knowledge, the lower bound of the expressivity for the VQA-based model has not been explored. With the aim of narrowing this knowledge gap, in the following, we conduct numerical simulations to empirically understand the tightness of the derived bound in Theorem 1. }

\medskip
{ 
Here we empirically explore whether the derived upper bound in Theorem 1 is tight with respect to the number of parameters $N_{gt}$. Note that directly calculating the covering number $\mathcal{N}((\HQNN)_{|S},  \epsilon, \|\cdot\|_2)$ is very challenged in general. To this end, we devise an alternative method to examine the tightness of our bounds. Recall that the main conclusion achieved in Theorem 2 is 
\begin{eqnarray}
	&& \mathcal{R}(\mathcal{A}(S))-\hat{\mathcal{R}}_S(\mathcal{A}(S)) 
	\nonumber\\
\leq && 2L_1 \inf_{\alpha>0}\left(4\alpha + \frac{12}{\sqrt{n}}\int_{\alpha}^{1}\sqrt{\ln \mathcal{N}((\HQNN)_{|S},  \epsilon, \|\cdot\|_2)} d \epsilon \right) +  3C_1\sqrt{\frac{\ln(2/\delta)}{2n}} \nonumber\\
\leq && \frac{8L_1}{\sqrt{n}} +  \frac{24L_1}{\sqrt{n}} d^{k}\sqrt{N_{gt}} \left(\ln\left(7 \sqrt{n} N_{gt} \|O\| \right)+1\right) +   3C_1\sqrt{\frac{\ln(1/\delta)}{2n}}. \nonumber
\end{eqnarray}
This result connects the generalization error $\mathcal{R}(\mathcal{A}(S))-\hat{\mathcal{R}}_S(\mathcal{A}(S))$ with the expressivity of QNN, i.e., $\mathcal{N}((\HQNN)_{|S},  \epsilon, \|\cdot\|_2)$. With this regard,  the  quantification of the tightness of the derived upper bound amounts to examining whether the generalization error of QNNs is linearly scaled with $\sqrt{N_{gt}}$. Mathematically, we aim to observe the relation $\mathcal{R}(\mathcal{A}(S))-\hat{\mathcal{R}}_S(\mathcal{A}(S))\sim O(\sqrt{N_{gt}})$  to validate the tightness of our bounds in terms of $N_{gt}$.
}
 
\begin{figure*}
	\centering
\includegraphics[width=0.65\textwidth]{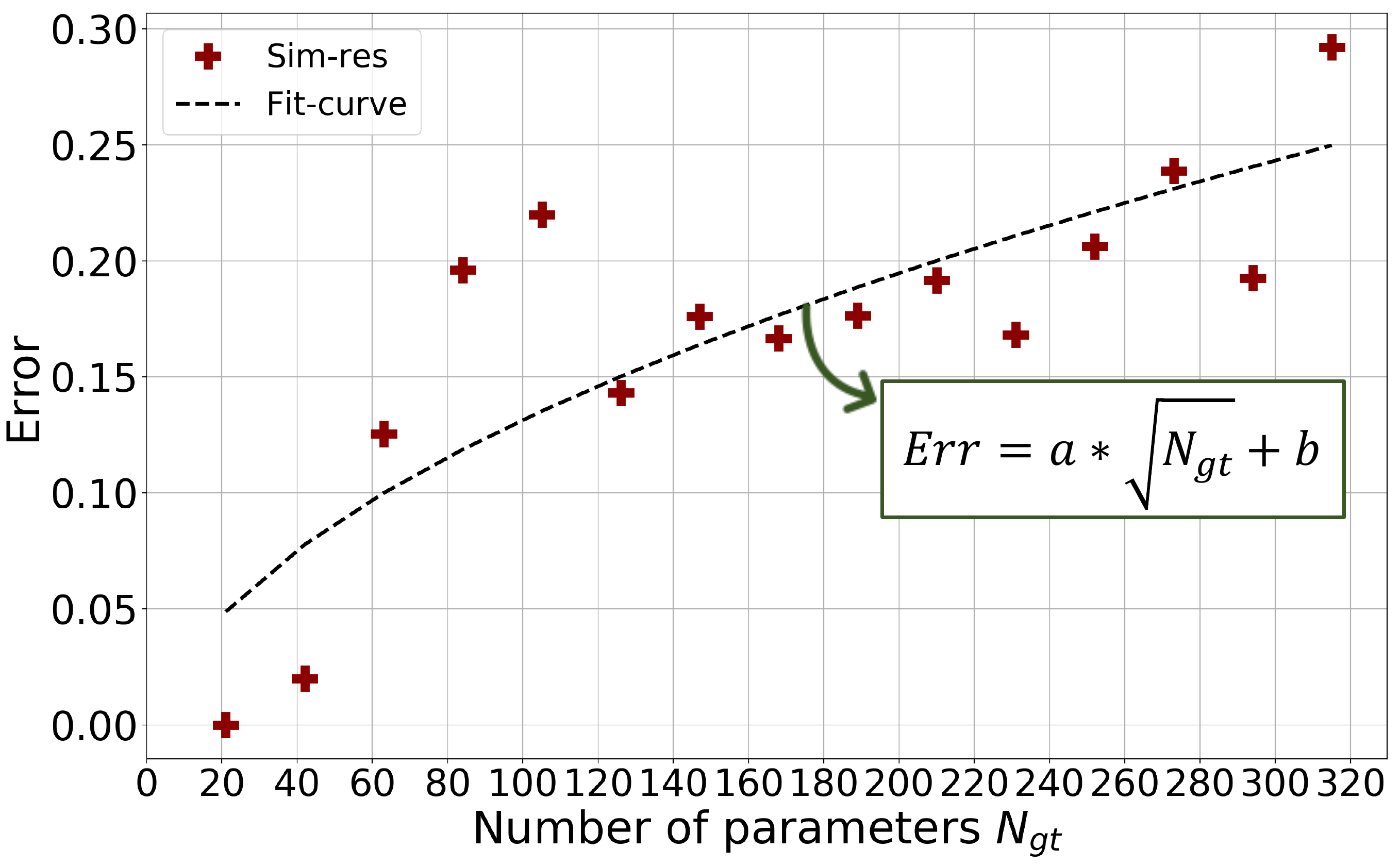}
\caption{\small{\textbf{The relationship between the learning error and the number of trainable parameters $N_{gt}$}.  The label `Sim-res’ refers to the averaged simulation results. The label `Fit-curve’ refers to the fitting curve, i.e., the mathematical form is $a \sqrt{N_{gt}} + b$ with $a, b \in \mathbb{R}$, with respect to the collected simulation results with the varied settings. }}
\label{fig:curve-fit}
\end{figure*}

We now employ the numerical simulations related to QNNs as introduced in the main text and Appendix \ref{append:QNN} to complete the above examination. Specifically, all hyper-parameters settings are identical to those introduced in the main text. The only modification is varying the layer number from $L=1$ to $L=15$. The simulations results are depicted in Fig.~\ref{fig:curve-fit}. Through fitting the simulation results, we observe that the learning error linearly scales with $\sqrt{N_{gt}}$, which accords with our theoretical results. \revise{We further use the coefficient of determination \cite{nagelkerke1991note}, denoted as $R^2\in [0,1]$, to measure the error of fitting. Intuitively, a higher $R^2$ reflects a good fitting curve, where $R^2=1$ indicates that the model explains all the variability of the response data around its mean. The coefficient of determination for the fitting curve shown in Fig.~\ref{fig:curve-fit} yields  $R^2=0.6687$.} These observations provide concrete evidence that the derived bound is tight with respect to the number of parameters.


\begin{thebibliography}{10}

\bibitem{biamonte2017quantum}
Jacob Biamonte, Peter Wittek, Nicola Pancotti, Patrick Rebentrost, Nathan
  Wiebe, and Seth Lloyd.
\newblock Quantum machine learning.
\newblock {\em Nature}, 549(7671):195, 2017.

\bibitem{dunjko2018machine}
Vedran Dunjko and Hans~J Briegel.
\newblock Machine learning \& artificial intelligence in the quantum domain: a
  review of recent progress.
\newblock {\em Reports on Progress in Physics}, 81(7):074001, 2018.

\bibitem{harrow2017quantum}
Aram~W Harrow and Ashley Montanaro.
\newblock Quantum computational supremacy.
\newblock {\em Nature}, 549(7671):203, 2017.

\bibitem{benedetti2019parameterized}
Marcello Benedetti, Erika Lloyd, Stefan Sack, and Mattia Fiorentini.
\newblock Parameterized quantum circuits as machine learning models.
\newblock {\em Quantum Science and Technology}, 4(4):043001, 2019.

\bibitem{bharti2021noisy}
Kishor Bharti, Alba Cervera-Lierta, Thi~Ha Kyaw, Tobias Haug, Sumner
  Alperin-Lea, Abhinav Anand, Matthias Degroote, Hermanni Heimonen, Jakob~S
  Kottmann, Tim Menke, et~al.
\newblock Noisy intermediate-scale quantum (nisq) algorithms.
\newblock {\em arXiv preprint arXiv:2101.08448}, 2021.

\bibitem{cerezo2020variational2}
M~Cerezo, Andrew Arrasmith, Ryan Babbush, Simon~C Benjamin, Suguru Endo,
  Keisuke Fujii, Jarrod~R McClean, Kosuke Mitarai, Xiao Yuan, Lukasz Cincio,
  et~al.
\newblock Variational quantum algorithms.
\newblock {\em arXiv preprint arXiv:2012.09265}, 2020.

\bibitem{du2018expressive}
Yuxuan Du, Min-Hsiu Hsieh, Tongliang Liu, and Dacheng Tao.
\newblock Expressive power of parametrized quantum circuits.
\newblock {\em Phys. Rev. Research}, 2:033125, Jul 2020.

\bibitem{endo2021hybrid}
Suguru Endo, Zhenyu Cai, Simon~C Benjamin, and Xiao Yuan.
\newblock Hybrid quantum-classical algorithms and quantum error mitigation.
\newblock {\em Journal of the Physical Society of Japan}, 90(3):032001, 2021.

\bibitem{preskill2018quantum}
John Preskill.
\newblock Quantum computing in the nisq era and beyond.
\newblock {\em Quantum}, 2:79, 2018.

\bibitem{Du_2021_grover}
Yuxuan Du, Min-Hsiu Hsieh, Tongliang Liu, and Dacheng Tao.
\newblock A grover-search based quantum learning scheme for classification.
\newblock {\em New Journal of Physics}, 23(2):023020, feb 2021.

\bibitem{du2020learnability}
Yuxuan Du, Min-Hsiu Hsieh, Tongliang Liu, Shan You, and Dacheng Tao.
\newblock On the learnability of quantum neural networks.
\newblock {\em arXiv preprint arXiv:2007.12369}, 2020.

\bibitem{huang2021information}
Hsin-Yuan Huang, Richard Kueng, and John Preskill.
\newblock Information-theoretic bounds on quantum advantage in machine
  learning.
\newblock {\em Physical Review Letters}, 126(19):190505, 2021.

\bibitem{shen2019information}
Huitao Shen, Pengfei Zhang, Yi-Zhuang You, and Hui Zhai.
\newblock Information scrambling in quantum neural networks.
\newblock {\em Physical Review Letters}, 124(20):200504, 2020.

\bibitem{wu2021expressivity}
Yadong Wu, Juan Yao, Pengfei Zhang, and Hui Zhai.
\newblock Expressivity of quantum neural networks.
\newblock {\em arXiv preprint arXiv:2101.04273}, 2021.

\bibitem{beer2020training}
Kerstin Beer, Dmytro Bondarenko, Terry Farrelly, Tobias~J Osborne, Robert
  Salzmann, Daniel Scheiermann, and Ramona Wolf.
\newblock Training deep quantum neural networks.
\newblock {\em Nature Communications}, 11(1):1--6, 2020.

\bibitem{mitarai2018quantum}
Kosuke Mitarai, Makoto Negoro, Masahiro Kitagawa, and Keisuke Fujii.
\newblock Quantum circuit learning.
\newblock {\em Physical Review A}, 98(3):032309, 2018.

\bibitem{havlivcek2019supervised}
Vojt{\v{e}}ch Havl{\'\i}{\v{c}}ek, Antonio~D C{\'o}rcoles, Kristan Temme,
  Aram~W Harrow, Abhinav Kandala, Jerry~M Chow, and Jay~M Gambetta.
\newblock Supervised learning with quantum-enhanced feature spaces.
\newblock {\em Nature}, 567(7747):209, 2019.

\bibitem{wang2020quantum}
Kunkun Wang, Lei Xiao, Wei Yi, Shi-Ju Ran, and Peng Xue.
\newblock Quantum image classifier with single photons.
\newblock {\em arXiv preprint arXiv:2003.08551}, 2020.

\bibitem{peruzzo2014variational}
Alberto Peruzzo, Jarrod McClean, Peter Shadbolt, Man-Hong Yung, Xiao-Qi Zhou,
  Peter~J Love, Al{\'a}n Aspuru-Guzik, and Jeremy~L O’brien.
\newblock A variational eigenvalue solver on a photonic quantum processor.
\newblock {\em Nature communications}, 5:4213, 2014.

\bibitem{kandala2017hardware}
Abhinav Kandala, Antonio Mezzacapo, Kristan Temme, Maika Takita, Markus Brink,
  Jerry~M Chow, and Jay~M Gambetta.
\newblock Hardware-efficient variational quantum eigensolver for small
  molecules and quantum magnets.
\newblock {\em Nature}, 549(7671):242--246, 2017.

\bibitem{google2020hartree}
Google~AI Quantum et~al.
\newblock Hartree-fock on a superconducting qubit quantum computer.
\newblock {\em Science}, 369(6507):1084--1089, 2020.

\bibitem{huang2020experimental}
He-Liang Huang, Yuxuan Du, Ming Gong, Youwei Zhao, Yulin Wu, Chaoyue Wang,
  Shaowei Li, Futian Liang, Jin Lin, Yu~Xu, et~al.
\newblock Experimental quantum generative adversarial networks for image
  generation.
\newblock {\em arXiv preprint arXiv:2010.06201}, 2020.

\bibitem{zhu2019training}
Daiwei Zhu, Norbert~M Linke, Marcello Benedetti, Kevin~A Landsman, Nhung~H
  Nguyen, C~Huerta Alderete, Alejandro Perdomo-Ortiz, Nathan Korda, A~Garfoot,
  Charles Brecque, et~al.
\newblock Training of quantum circuits on a hybrid quantum computer.
\newblock {\em Science advances}, 5(10):eaaw9918, 2019.

\bibitem{cerezo2020cost}
M~Cerezo, Akira Sone, Tyler Volkoff, Lukasz Cincio, and Patrick~J Coles.
\newblock Cost-function-dependent barren plateaus in shallow quantum neural
  networks.
\newblock {\em arXiv preprint arXiv:2001.00550}, 2020.

\bibitem{mcclean2018barren}
Jarrod~R McClean, Sergio Boixo, Vadim~N Smelyanskiy, Ryan Babbush, and Hartmut
  Neven.
\newblock Barren plateaus in quantum neural network training landscapes.
\newblock {\em Nature communications}, 9(1):1--6, 2018.

\bibitem{wang2020noise}
Samson Wang, Enrico Fontana, Marco Cerezo, Kunal Sharma, Akira Sone, Lukasz
  Cincio, and Patrick~J Coles.
\newblock Noise-induced barren plateaus in variational quantum algorithms.
\newblock {\em arXiv preprint arXiv:2007.14384}, 2020.

\bibitem{zhang2020toward}
Kaining Zhang, Min-Hsiu Hsieh, Liu Liu, and Dacheng Tao.
\newblock Toward trainability of quantum neural networks.
\newblock {\em arXiv preprint arXiv:2011.06258}, 2020.

\bibitem{sweke2020stochastic}
Ryan Sweke, Frederik Wilde, Johannes Meyer, Maria Schuld, Paul~K F{\"a}hrmann,
  Barth{\'e}l{\'e}my Meynard-Piganeau, and Jens Eisert.
\newblock Stochastic gradient descent for hybrid quantum-classical
  optimization.
\newblock {\em Quantum}, 4:314, 2020.

\bibitem{abbas2020power}
Amira Abbas, David Sutter, Christa Zoufal, Aur{\'e}lien Lucchi, Alessio
  Figalli, and Stefan Woerner.
\newblock The power of quantum neural networks.
\newblock {\em arXiv preprint arXiv:2011.00027}, 2020.

\bibitem{banchi2021generalization}
Leonardo Banchi, Jason Pereira, and Stefano Pirandola.
\newblock Generalization in quantum machine learning: a quantum information
  perspective.
\newblock {\em arXiv preprint arXiv:2102.08991}, 2021.

\bibitem{bu2021statistical}
Kaifeng Bu, Dax~Enshan Koh, Lu~Li, Qingxian Luo, and Yaobo Zhang.
\newblock On the statistical complexity of quantum circuits.
\newblock {\em arXiv preprint arXiv:2101.06154}, 2021.

\bibitem{caro2020pseudo}
Matthias~C Caro and Ishaun Datta.
\newblock Pseudo-dimension of quantum circuits.
\newblock {\em Quantum Machine Intelligence}, 2(2):1--14, 2020.

\bibitem{funcke2021dimensional}
Lena Funcke, Tobias Hartung, Karl Jansen, Stefan K{\"u}hn, and Paolo Stornati.
\newblock Dimensional expressivity analysis of parametric quantum circuits.
\newblock {\em Quantum}, 5:422, 2021.

\bibitem{poland2020no}
Kyle Poland, Kerstin Beer, and Tobias~J Osborne.
\newblock No free lunch for quantum machine learning.
\newblock {\em arXiv preprint arXiv:2003.14103}, 2020.

\bibitem{holmes2021connecting}
Zo{\"e} Holmes, Kunal Sharma, M~Cerezo, and Patrick~J Coles.
\newblock Connecting ansatz expressibility to gradient magnitudes and barren
  plateaus.
\newblock {\em arXiv preprint arXiv:2101.02138}, 2021.

\bibitem{sim2019expressibility}
Sukin Sim, Peter~D Johnson, and Al{\'a}n Aspuru-Guzik.
\newblock Expressibility and entangling capability of parameterized quantum
  circuits for hybrid quantum-classical algorithms.
\newblock {\em Advanced Quantum Technologies}, 2(12):1900070, 2019.

\bibitem{nakaji2021expressibility}
Kouhei Nakaji and Naoki Yamamoto.
\newblock Expressibility of the alternating layered ansatz for quantum
  computation.
\newblock {\em Quantum}, 5:434, 2021.

\bibitem{harrow2009random}
Aram~W Harrow and Richard~A Low.
\newblock Random quantum circuits are approximate 2-designs.
\newblock {\em Communications in Mathematical Physics}, 291(1):257--302, 2009.

\bibitem{huggins2019towards}
William Huggins, Piyush Patil, Bradley Mitchell, K~Birgitta Whaley, and E~Miles
  Stoudenmire.
\newblock Towards quantum machine learning with tensor networks.
\newblock {\em Quantum Science and technology}, 4(2):024001, 2019.

\bibitem{vapnik2013nature}
Vladimir Vapnik.
\newblock {\em The nature of statistical learning theory}.
\newblock Springer science \& business media, 2013.

\bibitem{schuld2019evaluating}
Maria Schuld, Ville Bergholm, Christian Gogolin, Josh Izaac, and Nathan
  Killoran.
\newblock Evaluating analytic gradients on quantum hardware.
\newblock {\em Physical Review A}, 99(3):032331, 2019.

\bibitem{stokes2020quantum}
James Stokes, Josh Izaac, Nathan Killoran, and Giuseppe Carleo.
\newblock Quantum natural gradient.
\newblock {\em Quantum}, 4:269, 2020.

\bibitem{poulin2011quantum}
David Poulin, Angie Qarry, Rolando Somma, and Frank Verstraete.
\newblock Quantum simulation of time-dependent hamiltonians and the convenient
  illusion of hilbert space.
\newblock {\em Physical review letters}, 106(17):170501, 2011.

\bibitem{szarek2010often}
Stanis{\l}aw~J Szarek, Elisabeth Werner, and Karol {\.Z}yczkowski.
\newblock How often is a random quantum state k-entangled?
\newblock {\em Journal of Physics A: Mathematical and Theoretical},
  44(4):045303, 2010.

\bibitem{mohri2018foundations}
Mehryar Mohri, Afshin Rostamizadeh, and Ameet Talwalkar.
\newblock {\em Foundations of machine learning}.
\newblock MIT press, 2018.

\bibitem{Note1}
{We remark that to obtain a general result that covers any type of noise, a
  relatively loose relaxation technique is used to infer $\protect \mathcal
  {N}(\protect \mathaccentV {tilde}07E{\protect \mathcal {H}}, \epsilon ,
  |\cdot |)$. This leads to a different scaling behavior in term of $\delimiter
  "026B30D O\delimiter "026B30D $ comparing with the ideal and depolarizing
  cases.}

\bibitem{cai2020mitigating}
Zhenyu Cai, Xiaosi Xu, and Simon~C Benjamin.
\newblock Mitigating coherent noise using pauli conjugation.
\newblock {\em npj Quantum Information}, 6(1):1--9, 2020.

\bibitem{du2020quantum}
Yuxuan Du, Tao Huang, Shan You, Min-Hsiu Hsieh, and Dacheng Tao.
\newblock Quantum circuit architecture search: error mitigation and
  trainability enhancement for variational quantum solvers.
\newblock {\em arXiv preprint arXiv:2010.10217}, 2020.

\bibitem{mcardle2019error}
Sam McArdle, Xiao Yuan, and Simon Benjamin.
\newblock Error-mitigated digital quantum simulation.
\newblock {\em Physical review letters}, 122(18):180501, 2019.

\bibitem{mcclean2020decoding}
Jarrod~R McClean, Zhang Jiang, Nicholas~C Rubin, Ryan Babbush, and Hartmut
  Neven.
\newblock Decoding quantum errors with subspace expansions.
\newblock {\em Nature communications}, 11(1):1--9, 2020.

\bibitem{strikis2020learning}
Armands Strikis, Dayue Qin, Yanzhu Chen, Simon~C Benjamin, and Ying Li.
\newblock Learning-based quantum error mitigation.
\newblock {\em arXiv preprint arXiv:2005.07601}, 2020.

\bibitem{sun2021mitigating}
Jinzhao Sun, Xiao Yuan, Takahiro Tsunoda, Vlatko Vedral, Simon~C Benjamin, and
  Suguru Endo.
\newblock Mitigating realistic noise in practical noisy intermediate-scale
  quantum devices.
\newblock {\em Physical Review Applied}, 15(3):034026, 2021.

\bibitem{kawaguchi2017generalization}
Kenji Kawaguchi, Leslie~Pack Kaelbling, and Yoshua Bengio.
\newblock Generalization in deep learning.
\newblock {\em arXiv preprint arXiv:1710.05468}, 2017.

\bibitem{larose2020robust}
Ryan LaRose and Brian Coyle.
\newblock Robust data encodings for quantum classifiers.
\newblock {\em Physical Review A}, 102(3):032420, 2020.

\bibitem{mohri2012foundations}
Mehryar Mohri, Afshin Rostamizadeh, and Ameet Talwalkar.
\newblock Foundations of machine learning, 2012.

\bibitem{marrero2020entanglement}
Carlos~Ortiz Marrero, M{\'a}ria Kieferov{\'a}, and Nathan Wiebe.
\newblock Entanglement induced barren plateaus.
\newblock {\em arXiv preprint arXiv:2010.15968}, 2020.

\bibitem{sweke2019stochastic}
Ryan Sweke, Frederik Wilde, Johannes Meyer, Maria Schuld, Paul~K F{\"a}hrmann,
  Barth{\'e}l{\'e}my Meynard-Piganeau, and Jens Eisert.
\newblock Stochastic gradient descent for hybrid quantum-classical
  optimization.
\newblock {\em Quantum}, 4:314, 2020.

\bibitem{blumer1987occam}
Anselm Blumer, Andrzej Ehrenfeucht, David Haussler, and Manfred~K Warmuth.
\newblock Occam's razor.
\newblock {\em Information processing letters}, 24(6):377--380, 1987.

\bibitem{larocca2021theory}
Martin Larocca, Nathan Ju, Diego García-Martín, Patrick~J. Coles, and
  M.~Cerezo.
\newblock Theory of overparametrization in quantum neural networks, 2021.

\bibitem{canatar2021spectral}
Abdulkadir Canatar, Blake Bordelon, and Cengiz Pehlevan.
\newblock Spectral bias and task-model alignment explain generalization in
  kernel regression and infinitely wide neural networks.
\newblock {\em Nature communications}, 12(1):1--12, 2021.

\bibitem{mcardle2020quantum}
Sam McArdle, Suguru Endo, Alan Aspuru-Guzik, Simon~C Benjamin, and Xiao Yuan.
\newblock Quantum computational chemistry.
\newblock {\em Reviews of Modern Physics}, 92(1):015003, 2020.

\bibitem{cao2019quantum}
Yudong Cao, Jonathan Romero, Jonathan~P Olson, Matthias Degroote, Peter~D
  Johnson, M{\'a}ria Kieferov{\'a}, Ian~D Kivlichan, Tim Menke, Borja
  Peropadre, Nicolas~PD Sawaya, et~al.
\newblock Quantum chemistry in the age of quantum computing.
\newblock {\em Chemical reviews}, 119(19):10856--10915, 2019.

\bibitem{Note2}
The separated expressivity of different ansatze is completed by controlling the
  involved number of quantum gates, supported by Theorem~\ref
  {thm:main-cov-num}. See Appendix G for construction details.

\bibitem{Barthel2018fundamental}
Thomas Barthel and Jianfeng Lu.
\newblock Fundamental limitations for measurements in quantum many-body
  systems.
\newblock {\em Phys. Rev. Lett.}, 121:080406, Aug 2018.

\bibitem{nielsen2010quantum}
Michael~A Nielsen and Isaac~L Chuang.
\newblock {\em Quantum computation and quantum information}.
\newblock Cambridge University Press, 2010.

\bibitem{Kakade2008OnTC}
Sham~M. Kakade, K.~Sridharan, and Ambuj Tewari.
\newblock On the complexity of linear prediction: Risk bounds, margin bounds,
  and regularization.
\newblock In {\em NIPS}, 2008.

\bibitem{dudley1967sizes}
Richard~M Dudley.
\newblock The sizes of compact subsets of hilbert space and continuity of
  gaussian processes.
\newblock {\em Journal of Functional Analysis}, 1(3):290--330, 1967.

\bibitem{romero2018strategies}
Jonathan Romero, Ryan Babbush, Jarrod~R McClean, Cornelius Hempel, Peter~J
  Love, and Al{\'a}n Aspuru-Guzik.
\newblock Strategies for quantum computing molecular energies using the unitary
  coupled cluster ansatz.
\newblock {\em Quantum Science and Technology}, 4(1):014008, 2018.

\bibitem{bravyi2002fermionic}
Sergey~B Bravyi and Alexei~Yu Kitaev.
\newblock Fermionic quantum computation.
\newblock {\em Annals of Physics}, 298(1):210--226, 2002.

\bibitem{mcclean2020openfermion}
Jarrod~R McClean, Nicholas~C Rubin, Kevin~J Sung, Ian~D Kivlichan, Xavier
  Bonet-Monroig, Yudong Cao, Chengyu Dai, E~Schuyler Fried, Craig Gidney,
  Brendan Gimby, et~al.
\newblock Openfermion: the electronic structure package for quantum computers.
\newblock {\em Quantum Science and Technology}, 5(3):034014, 2020.

\bibitem{Kingma2014Adam}
Diederik~P. Kingma and Jimmy Ba.
\newblock Adam: {A} method for stochastic optimization.
\newblock In {\em 3rd International Conference on Learning Representations,
  {ICLR} 2015, San Diego, CA, USA, May 7-9, 2015, Conference Track
  Proceedings}, 2015.

\bibitem{caro2021encoding}
Matthias~C Caro, Elies Gil-Fuster, Johannes~Jakob Meyer, Jens Eisert, and Ryan
  Sweke.
\newblock Encoding-dependent generalization bounds for parametrized quantum
  circuits.
\newblock {\em arXiv preprint arXiv:2106.03880}, 2021.

\bibitem{gyurik2021structural}
Casper Gyurik, Dyon van Vreumingen, and Vedran Dunjko.
\newblock Structural risk minimization for quantum linear classifiers.
\newblock {\em arXiv preprint arXiv:2105.05566}, 2021.

\bibitem{bilkis2021semi}
M~Bilkis, M~Cerezo, Guillaume Verdon, Patrick~J Coles, and Lukasz Cincio.
\newblock A semi-agnostic ansatz with variable structure for quantum machine
  learning.
\newblock {\em arXiv preprint arXiv:2103.06712}, 2021.

\bibitem{grimsley2019adaptive}
Harper~R Grimsley, Sophia~E Economou, Edwin Barnes, and Nicholas~J Mayhall.
\newblock An adaptive variational algorithm for exact molecular simulations on
  a quantum computer.
\newblock {\em Nature communications}, 10(1):1--9, 2019.

\bibitem{kuo2021quantum}
En-Jui Kuo, Yao-Lung~L Fang, and Samuel Yen-Chi Chen.
\newblock Quantum architecture search via deep reinforcement learning.
\newblock {\em arXiv preprint arXiv:2104.07715}, 2021.

\bibitem{ostaszewski2019quantum}
Mateusz Ostaszewski, Edward Grant, and Marcello Benedetti.
\newblock Structure optimization for parameterized quantum circuits.
\newblock {\em Quantum}, 5:391, 2021.

\bibitem{tang2021qubit}
Ho~Lun Tang, VO~Shkolnikov, George~S Barron, Harper~R Grimsley, Nicholas~J
  Mayhall, Edwin Barnes, and Sophia~E Economou.
\newblock qubit-adapt-vqe: An adaptive algorithm for constructing
  hardware-efficient ans{\"a}tze on a quantum processor.
\newblock {\em PRX Quantum}, 2(2):020310, 2021.

\bibitem{zhang2021neural}
Shi-Xin Zhang, Chang-Yu Hsieh, Shengyu Zhang, and Hong Yao.
\newblock Neural predictor based quantum architecture search.
\newblock {\em arXiv preprint arXiv:2103.06524}, 2021.

\bibitem{pouyanfar2018survey}
Samira Pouyanfar, Saad Sadiq, Yilin Yan, Haiman Tian, Yudong Tao, Maria~Presa
  Reyes, Mei-Ling Shyu, Shu-Ching Chen, and Sundaraja~S Iyengar.
\newblock A survey on deep learning: Algorithms, techniques, and applications.
\newblock {\em ACM Computing Surveys (CSUR)}, 51(5):1--36, 2018.

\bibitem{sun2019optimization}
Ruoyu Sun.
\newblock Optimization for deep learning: theory and algorithms.
\newblock {\em arXiv preprint arXiv:1912.08957}, 2019.

\bibitem{bartlett2017spectrally}
Peter Bartlett, Dylan Foster, and Matus Telgarsky.
\newblock Spectrally-normalized margin bounds for neural networks.
\newblock {\em Advances in Neural Information Processing Systems},
  30:6241--6250, 2017.

\bibitem{nagelkerke1991note}
Nico~JD Nagelkerke et~al.
\newblock A note on a general definition of the coefficient of determination.
\newblock {\em Biometrika}, 78(3):691--692, 1991.

\end{thebibliography}
\end{document}